\documentclass{lmcs}
\pdfoutput=1
\usepackage[utf8]{inputenc}

\usepackage{lastpage}
\lmcsdoi{19}{3}{15}
\lmcsheading{}{\pageref{LastPage}}{}{}%
{Dec.~28,~2022}{Sep.~04,~2023}{}

\usepackage{default}

\begin{document}
\title{Monoidal Width}
\author[E. Di Lavore]{Elena {Di Lavore}\lmcsorcid{0000-0002-7783-5079}}
\author[P. Soboci\'nski]{{Pawe\l} Soboci\'nski\lmcsorcid{0000-0002-7992-9685}}

\address{\begin{center}Tallinn University of Technology, Tallinn, Estonia\end{center}}
\thanks{This research was supported by the ESF funded Estonian IT Academy research measure (project 2014-2020.4.05.19-0001).
The second author was additionally supported by the Estonian Research Council grant PRG1210.}
\begin{abstract}
  We introduce monoidal width as a measure of complexity for morphisms in monoidal categories.
  Inspired by well-known structural width measures for graphs, like tree width and rank width, monoidal width is based on a notion of syntactic decomposition:
  a monoidal decomposition of a morphism is an expression in the language of monoidal categories, where operations are monoidal products and compositions, that specifies this morphism.
  Monoidal width penalises the composition operation along ``big'' objects, while it encourages the use of monoidal products.
  We show that, by choosing the correct categorical algebra for decomposing graphs, we can capture tree width and rank width.
  For matrices, monoidal width is related to the rank.
  These examples suggest monoidal width as a good measure for structural complexity of processes modelled as morphisms in monoidal categories.
\end{abstract}
\maketitle
\tableofcontents
\section{Introduction}
% structure meets power
In recent years, a current of research has emerged with focus on the interaction of \emph{structure} --- especially algebraic, using category theory and related subjects --- and \emph{power}, that is algorithmic and combinatorial insights stemming from graph theory, game theory and related subjects. Recent works include~\cite{abramsky2017pebbling,abramsky2021relating,montacute2022pebble}.

% algebra of monoidal cats and applications
The algebra of monoidal categories is a fruitful source of \emph{structure} --- it can be seen as a general process algebra of concurrent processes, featuring a sequential ($\dcomp$) as well as a parallel ($\tensor$) composition. Serving as a process algebra in this sense, it has been used to describe artefacts of a computational nature as \emph{arrows} of appropriate monoidal categories. Examples include Petri nets~\cite{fong2018seven}, quantum circuits~\cite{Coecke2017,DuncanKPW20},
signal flow graphs~\cite{fong2018seven,Bonchi0Z21}, electrical circuits~\cite{Comfort2021,Boisseau2021}, digital circuits~\cite{GhicaJL17}, stochastic processes~\cite{fritz2020,cho2019} and games~\cite{GhaniHWZ18}.

% divide-and-conquer algorithms as a motivation
Given that the algebra of monoidal categories has proved its utility as a language for describing computational artefacts in various applications areas, a natural question is to examine its relationship with \emph{power}: can monoidal structure help us to design efficient algorithms? To begin to answer this question, let us consider a mainstay of computer science: \emph{divide-and-conquer} algorithms. Such algorithms rely on the internal geometry of the global artefact under consideration to ensure the ability to \emph{divide}, that is, decompose it consistently into simpler components, inductively compute partial solutions on the components, and then recombine these local results to obtain a global solution.

\begin{figure}[h!]
  \interchangelawDecFig{}
  \caption{This morphism can be decomposed in two different ways: \((f \tensor f') \dcomp (g \tensor g') = (f \dcomp g) \tensor (f' \dcomp g')\).}\label{fig:interchange-law}
\end{figure}

% relating divide-and-conquer with monoidal stuff
Let us now return to systems described as arrows of monoidal categories. In applications, the parallel ($\tensor$) composition typically means placing systems side-by-side with no explicit interconnections. On the other hand, the sequential ($\dcomp$) composition along an object typically means communication, resource sharing or synchronisation, the complexity of which is determined by the object along which the composition is performed. Based on examples in the literature, our basic motivating intuition is:
\slogan{An algorithmic problem on an artefact that is a `$\tensor$' lends itself to a divide-and-conquer approach more easily than one that is a `$\dcomp$'.}
Moreover, the ``size'' of the object along which the `$\dcomp$' occurs matters; typically the ``larger'' the object, the more work is needed in order to recombine results in any kind of divide-and-conquer approach. An example is compositional reachability checking in Petri nets of Rathke et.\ al.~\cite{rathke2014compositional}: calculating the sequential composition is exponential in the size of the boundary. Another recent example is the work of Master~\cite{master2022compose} on a compositional approach to calculating shortest paths.

\medskip
% the idea of monoidal width
On the other hand, (monoidal) category theory equates different descriptions of systems. Consider what is known as middle-four interchange, illustrated in Figure~\ref{fig:interchange-law}.
Although monoidal category theory asserts that \((f \tensor f') \dcomp (g \tensor g') = (f \dcomp g) \tensor (f' \dcomp g')\), considering the two sides of the equations as decomposition blueprints for a divide-and-conquer approach, the right-hand side of the equation is clearly preferable since it maximises parallelism by minimising the size of the boundary along which composition occurs. This, roughly speaking, is the idea of \emph{width} -- expressions in the language of monoidal categories are assigned a natural number that measures ``how good'' they are as decomposition blueprints. The \emph{monoidal width} of an arrow is then the width of its most efficient decomposition. In concrete examples, arrows with low width lend themselves to efficient divide-and-conquer approaches, following a width-optimal expression as a decomposition blueprint.

\medskip

% highlight similarity with graph widths
The study of efficient decompositions of combinatorial artefacts is well-established, especially in graph theory.
A number of \emph{graph widths} --- by which we refer to related concepts like tree width, path with, branch width, cut width, rank width or twin width --- have become known in computer science because of their relationship with algorithmic properties.
All of them share a similar basic idea: in each case, a specific notion of legal decomposition is priced according to the most expensive operation involved, and the price of the cheapest decomposition is the width.

% tree width
Perhaps the most famous of these is tree width, a measure of complexity for graphs that was independently defined by different authors~\cite{bertele1973treewidth,halin1976treewidth,robertson1986graph-minorsII}.
Every nonempty graph has a tree width, which is a natural number.
Intuitively, a tree decomposition is a recipe for decomposing a graph into smaller subgraphs that form a tree shape.
These subgraphs, when some of their vertices are identified, need to compose into the original graph, as shown in \Cref{fig:ex-tree-dec}.
Courcelle's theorem
\slogan{Every property expressible in the monadic second order logic of graphs can be verified in linear time on graphs with bounded tree width.}
is probably the best known among several results that establish links with algorithms~\cite{bodlaender1992tourist,bodlaender2008combinatorial,courcelle1990monadic} thus illustrating its importance for computer science.
\begin{figure}[h!]
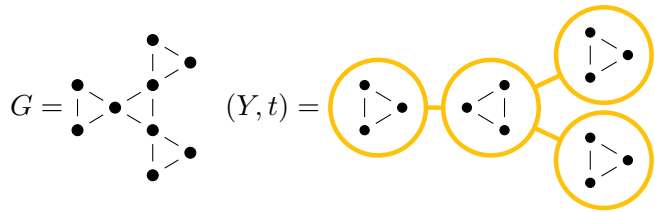

  \[\treeDecExFig{}\]
  \caption{A tree decomposition cuts the graph along its vertices.}\label{fig:ex-tree-dec}
\end{figure}

% rank width
Another important measure is rank width~\cite{oum2006rank-width} --- a relatively recent development that has attracted significant attention in the graph theory community.
A rank decomposition is a recipe for decomposing a graph into its single-vertex subgraphs by cutting along edges.
The cost of a cut is the rank of the adjacency matrix that represents it, as illustrated in \Cref{fig:ex-rank-cut}.
An intuition for rank width is that it is a kind of \virgolette{Kolmogorov complexity} for graphs, with higher rank widths indicating that the connectivity data of the graph cannot be easily compressed.
For example, while the family of cliques has unbounded tree width, their connectivity rather simple: in fact, all cliques have rank width $1$.
\begin{figure}[h!]
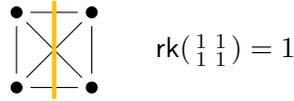

  \[\cutrankExFig{}\]
  \caption{A cut and its matrix in a rank decomposition.}\label{fig:ex-rank-cut}
\end{figure}

\paragraph{Contribution.}
% our goals in this paper
Building on our conference paper~\cite{2022monoidalwidth}, our goals are twofold.
Firstly, to introduce the concept of monoidal width and begin to develop techniques for reasoning about it.

\medskip
Before describing concrete, technical contributions, let us take a bird's eye view. It is natural for the seasoned researcher to be sceptical of a new abstract framework that seeks to generalise known results. The best abstract approaches \textit{(i)} simplify existing known arguments, \textit{(ii)} clean up the research landscape by connecting existing notions, or \textit{(iii)} introduce  techniques that allow one to prove new theorems. This paper does not (yet) bring strong arguments in favour of monoidal width if one uses these three points as yardsticks. Our high-level, conceptual contribution is, instead, the fact that the algebra of monoidal categories -- already used in several contexts in theoretical computer science -- is a multi-purpose algebra for specifying decompositions of graph-like structures important for computer scientists. There are several ways of making this work, and making these monoidal \emph{algebras} of ``open graphs'' explicit as \emph{monoidal categories} is itself a valuable endeavour.
Indeed, identifying a monoidal category automatically yields a particular notion of decomposition: the instance of monoidal width in the monoidal category of interest. This point of view therefore demystifies ad hoc notions of decomposition that accompany each notion of width that we consider in this paper.
Moreover, having an explicit algebra is also useful because it suggests a data structure --- the expression in the language of monoidal categories --- as a way of describing decompositions.

The results in this paper can  be seen as a ``sanity check'' of these general claims, but can also be seen as taking the first technical steps in order to build towards points \textit{(i)-(iii)} of the previous paragraph.
To this end we examine monoidal width in the presence of common structure, such as coherent comultiplication on objects, and in a foundational setting such as the monoidal category of matrices.
Secondly, connecting this approach with previous work, to examine graph widths through the prism of monoidal width.
The two widths we focus on are tree width and rank width.
We show that both can be seen as instances of monoidal width.
The interesting part of this endeavour is identifying the monoidal category, and thus the relevant ``decomposition algebra'' of interest.

Unlike the situation with graph widths, it does not make sense to talk about monoidal width per se, since it is dependent on the choice of underlying monoidal category and thus a particular ``decomposition algebra''.
The decomposition algebras that underlie tree and rank decompositions reflect their intuitive understanding.
For tree width, this is a cospan category whose morphisms represent graphs with vertex interfaces, while for rank width it is a category whose morphisms represent graphs with edge interfaces, with adjacency matrices playing the role of tracking connectivity information within a graph.
We show that the monoidal width of a morphism in these two categories is bounded, respectively, by the branch (\Cref{th:mwd-bwd}) and rank width (\Cref{th:mwd-rwd}) of the corresponding graph. In the first instance, this is enough to establish the connection between monoidal width and tree width, given that it is known that tree width and branch width are closely related.
A small technical innovation is the definition of intermediate inductive notions of branch (\Cref{def:rec-branch-decomposition}) and rank (\Cref{def:rec-rank-dec}) decompositions, equivalent to the original  definitions via ``global'' combinatorial notions of graph decomposition. The inductive presentations are closer in spirit to the inductive definition of monoidal decomposition, and allow us to give direct proofs of the main correspondences.

\paragraph{String diagrams.}
String diagrams~\cite{joyal1991geometry} are a convenient syntax for monoidal categories, where a morphism \(f \colon X \to Y\) is depicted as a box with input and output wires: \(\morphismFig{}\).
Morphisms in monoidal categories can be composed sequentially, using the composition of the category, and in parallel, using the monoidal structure.
These two kinds of composition are reflected in the string diagrammatic syntax:
the sequential composition \(f \dcomp g\) is depicted by connecting the output wire of \(f\) with the input wire of \(g\);
the parallel composition \(f \tensor f'\) is depicted by writing \(f\) on top of \(f'\).
\begin{align*}
  f \dcomp g &= \sequentialFig{} & f \tensor f' &= \parallelFig{}
\end{align*}
The advantage of this syntax is that all coherence equations for monoidal categories are trivially true when written with string diagrams.
An example is the middle-four interchange law \((f \tensor f') \dcomp (g \tensor g') = (f \dcomp g) \tensor (g \dcomp g')\).
These two expressions have one representation in terms of string diagrams, as shown in \Cref{fig:interchange-law}.
The coherence theorem for monoidal categories~\cite{maclane78} ensures that string diagrams are a sound and complete syntax for morphisms in monoidal categories.

\paragraph{Related work.}
This paper contains the results of~\cite{2021monoidalwidth} and~\cite{2022monoidalwidth} with detailed proofs.
We generalise the results of~\cite{2021monoidalwidth} to undirected hypergraphs and provide a syntactic presentation of the subcategory of the monoidal category of cospans of hypergraphs on discrete objects.

Previous syntactical approaches to graph widths are the work of Pudl{\'a}k, R{\"o}dl and Savick{\`y}~\cite{pudlak1988graph-complexity} and the work of Bauderon and Courcelle~\cite{bauderon1987graph}.
Their works consider different notions of graph decompositions, which lead to different notions of graph complexity.
In particular, in~\cite{bauderon1987graph}, the cost of a decomposition is measured by counting \emph{shared names}, which is clearly closely related to penalising sequential composition as in monoidal width.
Nevertheless, these approaches are specific to particular, concrete notions of graphs, whereas our work concerns the more general algebraic framework of monoidal categories.

Abstract approaches to width have received some attention recently, with a number of diverse contributions.
Blume et.\ al.~\cite{blume2011treewidth}, similarly to our work, use (the category of) cospans of graphs as a formal setting to study graph decompositions: indeed, a major insight of loc.\ cit.\ is that tree decompositions are tree-shaped diagrams in the cospan category, and the original graph is reconstructed as a colimit of such a diagram. Our approach is more general, however, emphasising the relevance of the algebra of monoidal categories, of which cospan categories are just one family of examples.

The literature on comonads for game semantics characterises tree and path decompositions of relational structures (and graphs in particular) as coalgebras of certain comonads~\cite{abramsky2017pebbling,abramsky2021relating,montacute2022pebble,abramsky2021comonadic,conghaile2021game}.
Bumpus and Kocsis~\cite{bumpus2021spined,bumpus2021generalizing} and, later, Bumpus, Kocsis and Master~\cite{bumpus2023structured} also generalise tree width to the categorical setting, although their approach is conceptually and technically removed from ours.
Their work takes a combinatorial perspective on decompositions, following the classical graph theory literature.
Given a shape of decomposition, called the spine in~\cite{bumpus2021spined}, a decomposition is defined globally as a functor out of that shape.
This generalises the characterisation of tree width based on Halin's S-functions~\cite{halin1976treewidth}.
In contrast, monoidal width is algebraic in flavour, following Bauderon and Courcelle's insights on tree decompositions~\cite{bauderon1987graph}.
Monoidal decompositions are syntax trees defined inductively and rely on the decomposition algebra given by monoidal categories.

\paragraph{Synopsis.}
The definition of monoidal width is introduced in \Cref{sec:mwd}, together with a worked out example.
In \Cref{sec:mwd-twd} we recover tree width by instantiating monoidal width in a suitable category of cospans of hypergraphs.
We recall it in \Cref{sec:cospans-hypergraphs} and provide a syntax for it in \Cref{sec:frobunnius-graphs}.
Similarly, in \Cref{sec:mwd-rwd} we recover rank width by instantiating monoidal width in a prop of graphs with boundaries where the connectivity information is stored in adjacency matrices, which we recall in \Cref{sec:prop-graph}.
This motivates us to study monoidal width for matrices over the natural numbers in \Cref{sec:mwd-matrices}.

\section{Monoidal width}\label{sec:mwd}
We introduce monoidal width, a notion of complexity for morphisms in monoidal categories that relies on explicit syntactic \emph{decompositions}, relying on the algebra of monoidal categories.
We then proceed with a simple, yet useful examples of efficient monoidal decompositions in~\Cref{sec:mwd-copy}.

A monoidal decomposition of a morphism \(f\) is a binary tree where internal nodes are labelled with the operations of composition \(\dcomp\) or monoidal product \(\tensor\), and leaves are labelled with ``atomic'' morphisms.
A decomposition, when evaluated in the obvious sense, results in \(f\).
We do not assume that the set of atomic morphisms \(\decgenerators\) is minimal, they are merely morphisms that do not necessarily need to be further decomposed.
We assume that \(\decgenerators\) contains enough atoms to have a decomposition for every morphism.
In most cases, we will take \(\decgenerators\) to contain all the morphisms.

\begin{defi}[Monoidal decomposition]\label{def:monoidal-decomposition}
  Let \(\cat{C}\) be a monoidal category and \(\decgenerators\) be a subset of its morphisms to which we refer as \emph{atomic}.
  The set \(\decset{f}\) of monoidal decompositions of \(f \colon A \to B\) in \(\cat{C}\) is defined inductively:
  \begin{align*}
    \decset{f} \quad \Coloneqq \quad & \leafgenerator{f} &\text{if } f \in \decgenerators \\
    \mid \quad & \nodegenerator{d_{1}}{\tensor}{d_{2}} &\text{if } d_1\in \decset{f_1},\,d_2\in \decset{f_2} \text{ and } f =_\cat{C} f_1\tensor f_2 \\
    \mid \quad & \nodegenerator{d_{1}}{\dcomp_{X}}{d_{2}} &\text{if }d_1\in \decset{f_1 \colon A\to X},\,d_2\in \decset{f_2 \colon X\to B} \text{ and }f =_\cat{C} f_1\dcomp f_2
  \end{align*}
\end{defi}

In general, a morphism can be decomposed in different ways and decompositions that maximise parallelism are deemed more efficient.
The monoidal width of a morphism is the cost of its cheapest monoidal decomposition.

Formally, each operation and atom in a decomposition is assigned a weight that will determine the cost of the decomposition. This is captured by the concept of a \emph{weight function}.
\begin{defi}\label{def:weight-function}
  Let \(\cat{C}\) be a monoidal category and let \(\decgenerators\) be its atomic morphisms.
  A \emph{weight function} for \((\cat{C},\decgenerators)\) is a function \(\nodeweight \colon \decgenerators \union \monoidaloperations{\cat{C}} \to \naturals\) such that \(\nodeweight(X \tensor Y) = \nodeweight(X) + \nodeweight(Y)\), and \(\nodeweight(\tensor) = 0\).
\end{defi}

A prop is a strict symmetric monoidal category where objects are natural numbers and the monoidal product on them is addition.
If \(\cat{C}\) is a prop, then, typically, we let \(\nodeweight(1)\defn 1\).
The idea behind giving a weight to an object $X\in\cat{C}$ is that $\nodeweight(X)$ is the cost paid for composing along $X$.

\begin{defi}[Monoidal width]\label{def:mwd}
  Let \(\nodeweight\) be a weight function for \((\cat{C},\decgenerators)\).
  Let \(f\) be in \(\cat{C}\) and \(d\in \decset{f}\).
  The \emph{width} of \(d\) is defined inductively as follows:
  \begin{align*}
    \decwidth(d) \defn\ &  \nodeweight(f) & \text{if }d=\leafgenerator{f} \\
                        &  \max \{\decwidth(d_1),\decwidth(d_2)\} & \text{if }d= \nodegenerator{d_{1}}{\tensor}{d_{2}}\\
                        &  \max \{\decwidth(d_1),\,\nodeweight(X),\,\decwidth(d_2)\} & \text{if } d= \nodegenerator{d_{1}}{\dcomp_{X}}{d_{2}}
  \end{align*}
  The \emph{monoidal width} of \(f\) is \(\mwd(f) \defn \min_{d\in \decset{f}} \decwidth(d)\).
\end{defi}

\begin{exa}\label{ex:monoidal-decomposition}
  Let \(f \colon 1 \to 2\) and \(g \colon 2 \to 1\) be morphisms in a prop such that \(\mwd(f)=\mwd(g)=2\).
  The following figure represents the monoidal decomposition of \(f \dcomp (f \tensor f) \dcomp (g \tensor g) \dcomp g\) given by
  \[\nodegenerator{f}{\dcomp_{2}}{\nodegenerator{\nodegenerator{\nodegenerator{f}{\dcomp_{2}}{g}}{\tensor}{\nodegenerator{f}{\dcomp_{2}}{g}}}{\dcomp_{2}}{g}}.\]
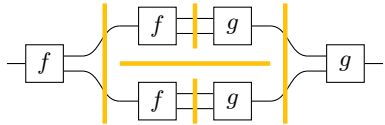
  \[\monoidaldecExFig{}\]
  Indeed, taking advantage of string diagrammatic syntax,
  decompositions can be illustrated by enhancing string diagrams with additional annotations that indicate the order of decomposition. Throughout this paper, we use thick yellow dividing lines for this purpose.

  Given that the width of a decomposition is the most expensive operation or atom, the above has width is \(2\) as compositions are along at most \(2\) wires.
\end{exa}

\begin{exa}\label{ex:mwd-number-like-morphisms}
  With the data of \Cref{ex:monoidal-decomposition}, define a family of morphisms \(h_{n} \colon 1 \to 1\) inductively as \(h_{0} \defn f \dcomp_{2} g\), and \(h_{n+1} \defn f \dcomp_{2} (h_{n} \tensor h_{n}) \dcomp_{2} g\).
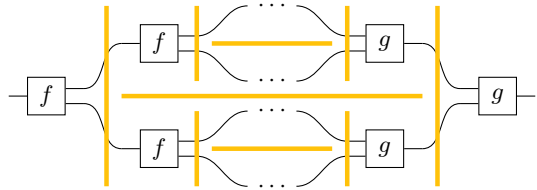
  \[\mwdNumberLikeExFig{}\]

  Each \(h_{n}\) has a decomposition of width \(2^{n}\) where the root node is the composition along the middle wires.
  However --- following the schematic diagram above --- we have that \(\mwd(h_{n}) \leq 2\) for any \(n\).
\end{exa}

\subsection{Monoidal width of copy}\label{sec:mwd-copy}

Although monoidal width is a very simple notion, reasoning about it in concrete examples can be daunting because of the combinatorial explosion in the number of possible decompositions of any morphism. For this reason, it is useful to examine some commonly occurring structures that one encounters ``in the wild'' and examine their decompositions. One such situation is when the objects are equipped with a coherent comultiplication structure.

\begin{defi}
  Let \(\cat{C}\) be a symmetric monoidal category, with symmetries \(\swap{X,Y} \colon X \tensor Y \to Y \tensor X\).
  We say that \(\cat{C}\) has coherent copying if there is a class of objects \(\copiable{\cat{C}} \subseteq \obj{\cat{C}}\), satisfying
  \begin{itemize}
    \item \(X, Y \in \copiable{\cat{C}}\) iff \(X \tensor Y \in \copiable{\cat{C}}\);
    \item Every object \(X \in \copiable{\cat{C}}\) is endowed with a morphism \(\cp_{X} \colon X \to X \tensor X\);
    \item For every \(X, Y \in \copiable{\cat{C}}\), \(\cp_{X \tensor Y} = (\cp_{X} \tensor \cp_{Y}) \dcomp (\id{X} \tensor \swap{X,Y} \tensor \id{Y})\) (coherence).
  \end{itemize}
  \[\coherenceCopyFig{}\]
\end{defi}
An example is any cartesian prop, where the copy morphisms are the universal ones given by the cartesian structure: \(\cp_n \defn \productmap{\id{n}}{\id{n}} \colon n \to n + n\).
For props with coherent copy, we assume that copy morphisms, symmetries and identities are atoms, \(\cp_{X},\swap{X,Y}, \id{X} \in \decgenerators\), and that their weight is given by \(\nodeweight(\cp_{X}) \defn 2 \cdot \nodeweight(X)\), \(\nodeweight(\swap{X,Y}) \defn \nodeweight(X) + \nodeweight(Y)\) and \(\nodeweight(\id{X}) \defn \nodeweight(X)\).

\begin{exa}
  Let \(\cat{C}\) be a prop with coherent copy and suppose that \(1 \in \copiable{\cat{C}}\).
  This implies that every \(n \in \copiable{\cat{C}}\) and there are copy morphisms \(\cp_{n} \colon n \to 2n\) for all \(n\).
  Let \(\gamma_{n,m} \defn (\cp_{n} \tensor \id{m}) \dcomp (\id{n} \tensor \swap{n,m}) \colon n + m \to n + m + n\).
  We can decompose \(\gamma_{n,m}\) in terms of \(\gamma_{n-1,m+1}\) (in the dashed box), \(\cp_{1}\) and \(\swap{1,1}\) by cutting along at most \(n+1+m\) wires:
  \[\gamma_{n,m} = (\id{n-1} \tensor ((\cp_{1} \tensor \id{1}) \dcomp (\id{1} \tensor \swap{1,1}))) \dcomp_{n+1+m} (g_{n-1,m+1} \tensor \id{1}).\]
  \[\gamma_{n,m} = \mwdCopyExFig{}\]
  This allows us to decompose \(\cp_{n} = \gamma_{n,0}\) cutting along only \(n+1\) wires. In particular, this means that \(\mwd(\cp_n) \leq n+1\).
\end{exa}

The following lemma generalises the above example and is used in the proofs of some results in later sections, \Cref{prop:mwd-bwd-upper} and \Cref{prop:mwd-lessthan-domain-codomain}.

\begin{lem}\label{lemma:mwd-copy}
  Let \(\cat{C}\) be a symmetric monoidal category with coherent copying.
  Suppose that \(\decgenerators\) contains \(\cp_X\) for all \(X \in \copiable{\cat{C}}\), and \(\swap{X,Y}\) and \(\id{X}\) for all \(X \in \obj{\cat{C}}\).
  Let \(\overline{X}\defn X_1 \tensor \cdots \tensor X_n\), with \(X_{i} \in \copiable{\cat{C}}\), \(f \colon Y \tensor \overline{X} \tensor Z \to W\) and let \(d \in \decset{f}\).
  Let \(\gamma_{\overline{X}}(f) \defn (\id{Y} \tensor \cp_{\overline{X}} \tensor \id{Z}) \dcomp (\id{Y \tensor \overline{X}} \tensor \swap{\overline{X}, Z}) \dcomp (f \tensor \id{\overline{X}})\).
  \[\gamma_{\overline{X}}(f) \,\defn\quad\lemmamwdcopyStateFig{}\]
  Then there is a monoidal decomposition \(\copyMdec_{\overline{X}}(d)\) of \(\gamma_{\overline{X}}(f)\) such that
  \[\decwidth(\copyMdec_{\overline{X}}(d)) \leq \max \{\decwidth(d), \nodeweight(Y) + \nodeweight(Z) + (n+1) \cdot \max_{i = 1,\ldots,n} \nodeweight(X_i)\}.\]
\end{lem}
\begin{proof}
  Proceed by induction on the number \(n\) of objects being copied.
  If \(n = 0\), then we are done because we keep the decomposition \(d\) and define \(\copyMdec_{\monoidalunit}(d) \defn d\).

  Suppose that the statement is true for any \(f' \colon Y \tensor \overline{X} \tensor Z' \to W\).
  Let \(f \colon Y \tensor \overline{X} \tensor X_{n+1} \tensor Z \to W\).
  By coherence of \(\cp\), we can rewrite \(\gamma_{\overline{X} \tensor X_{n+1}}(f)\).
  \begin{equation*}
    \lemmamwdcopyProofFigOne  = \lemmamwdcopyProofFigThree
  \end{equation*}
  Let \(\gamma_{\overline{X}}(f)\) be the morphism in the above dashed box.
  By the induction hypothesis, there is a monoidal decomposition \(\copyMdec_{\overline{X}}(d)\) of \(\gamma_{\overline{X}}(f)\) of bounded width: \(\decwidth(\copyMdec_{\overline{X}}(d)) \leq \max \{\decwidth(d), \nodeweight(Y) + \nodeweight(X_{n+1} \tensor Z) + (n+1) \cdot \max_{i = 1,\ldots,n} \nodeweight(X_i)\}\).
  We can use this decomposition to define a monoidal decomposition \(\copyMdec_{\overline{X} \tensor X_{n+1}}(d)\) of \(\gamma_{\overline{X}\tensor X_{n+1}}(f)\) as shown below.
  \[\lemmamwdcopyProofFigCuts{}\]
  Note that the only cut that matters is the longest vertical one, the composition node along \(Y \tensor \overline{X} \tensor X_{n+1} \tensor Z \tensor X_{n+1}\), because all the other cuts are cheaper.
  The cost of this cut is \(\nodeweight(Y)+\nodeweight(Z)+2 \cdot \nodeweight(X_{n+1}) + \nodeweight(\overline{X}) = \nodeweight(\dcompnode{Y}) + \nodeweight(\dcompnode{Z}) + \nodeweight(\dcompnode{X_{n+1}}) + \sum_{i=1}^{n+1} \nodeweight(\dcompnode{X_i})\).
  With this observation and applying the induction hypothesis, we can compute the width of the decomposition \(\copyMdec_{\overline{X} \tensor X_{n+1}}(d)\).
  \begin{align*}
    & \decwidth(\copyMdec_{\overline{X} \tensor X_{n+1}}(d)) \\
    & = \max \big\lbrace \nodeweight(\dcompnode{Y}) + \nodeweight(\dcompnode{Z}) + \nodeweight(\dcompnode{X_{n+1}}) + \sum_{i=1}^{n+1} \nodeweight(\dcompnode{X_i}), \decwidth(\copyMdec_{\overline{X}}(d))\big\rbrace \\
    & \leq \max\big\{\nodeweight(\dcompnode{Y}) + \nodeweight(\dcompnode{Z}) + (n+2) \cdot \max_{i=1, \ldots, n+1} \nodeweight(\dcompnode{X_i}), \decwidth(d),\\
    & \qquad \nodeweight(\dcompnode{Y}) + \nodeweight(\dcompnode{X_{n+1} \tensor Z}) + (n+1) \cdot \max_{i=1, \ldots, n} \nodeweight(\dcompnode{X_i})\big\} \\
    & = \max\big\{\nodeweight(\dcompnode{Y}) + \nodeweight(\dcompnode{Z}) + (n+2) \cdot \max_{i=1, \ldots, n+1} \nodeweight(\dcompnode{X_i}), \decwidth(d)\big\}
  \qedhere
  \end{align*}
\end{proof}

\section{A monoidal algebra for tree width}\label{sec:mwd-twd}
Our first case study is tree width of undirected hypergraphs.
We show that monoidal width in a suitable monoidal category of hypergraphs is within constant factors of tree width.
We rely on branch width, a measure equivalent to tree width, to relate the latter with monoidal width.

After recalling tree and branch width and the bounds between them in \Cref{sec:twd}, we define the intermediate notion of inductive branch decomposition in \Cref{sec:rec-branch-dec} and show its equivalence to that of branch decomposition.
Separating this intermediate step allows a clearer presentation of the correspondence between branch decompositions and monoidal decompositions.
\Cref{sec:cospans-hypergraphs} recalls the categorical algebra of cospans of hypergraphs and \Cref{sec:frobunnius-graphs} introduces a syntactic presentations of them.
Finally, \Cref{sec:mwd-twd-proof} contains the main result of the present section, which relates inductive branch decompositions, and thus tree decompositions, with monoidal decompositions.

Classically, tree and branch widths have been defined for finite undirected multihypergraphs, which we simply call hypergraphs.
These have undirected edges that connect sets of vertices and they may have parallel edges.

\begin{defi}\label{def:hypergraph}
  A \emph{(multi)hypergraph} \(G = \mathgraph{E}{V}\) is given by a finite set of vertices \(V\), a finite set of edges \(E\) and an adjacency function \(\edgeendsfun \colon E \to \parti(V)\), where \(\parti(V)\) indicates the set of subsets of \(V\).
  A \emph{subhypergraph} of \(G\) is a hypergraph \(G' = \mathgraph{E'}{V'}\) such that \(V' \subseteq V\), \(E' \subseteq E\) and \(\edgeendsfun'(e) = \edgeends{e}\) for all \(e \in E'\).
\end{defi}

\begin{defi}
  Given two hypergraphs \(G = \mathgraph{E}{V}\) and \(H = \mathgraph{F}{W}\), a \emph{hypergraph homomorphism} \(\alpha \colon G \to H\) is given by a pair of functions \(\alpha_{V} \colon V \to W\) and \(\alpha_{E} \colon E \to F\) such that, for all edges \(e \in E\), \(\edgesetends[H]{\alpha_{E}(e)} = \alpha_{V}(\edgesetends[G]{e})\).
  \[\graphmorphismDiagram{}\]
  Hypergraphs and hypergraph homomorphisms form a category \(\UHGraph\), where composition and identities are given by component-wise composition and identities.
\end{defi}

Note that the category \(\UHGraph\) is not the functor category \([\{\bullet \to \bullet\}, \kleisli{\parti}]\): their objects coincide but the morphisms are different.

\begin{defi}\label{def:hyperedge-size}
  The \emph{hyperedge size} of a hypergraph \(G\) is defined as $\hyperedgesize(G) \defn \max_{e \in \edges(G)}$ $\card{\edgeends{e}}$.
  A \emph{graph} is a hypergraph with hyperedge size \(2\).
\end{defi}

\begin{defi}
  A \emph{neighbour} of a vertex \(v\) is a vertex \(w\) distinct from \(v\) with an edge \(e\) such that \(v,w \in \edgeends{e}\).
  A \emph{path} in a hypergraph is a sequence of vertices \((v_{1}, \dots,v_{n})\) such that, for every \(i = 1,\dots,n-1\), \(v_{i}\) and \(v_{i+1}\) are neighbours.
  A \emph{cycle} in a hypergraph is a path where the first vertex \(v_{1}\) coincides with the last vertex \(v_{n}\).
  A hypergraph is \emph{connected} if there is a path between every two vertices.
  A \emph{tree} is a connected acyclic hypergraph.
  A tree is \emph{subcubic} if every vertex has at most three neighbours.
\end{defi}

\begin{defi}
  The set of \emph{binary trees} with labels in a set \(\Lambda\) is either: a leaf \(\leafgenerator{\lambda}\) with label \(\lambda \in \Lambda\); or a label \(\lambda \in \Lambda\) with two binary trees \(T_{1}\) and \(T_{2}\) with labels in \(\Lambda\), \(\nodegenerator{T_{1}}{\lambda}{T_{2}}\).
\end{defi}

\subsection{Background: tree width and branch width}\label{sec:twd}
Intuitively, tree width measures ``how far'' a hypergraph \(G\) is from being a tree: a hypergraph is a tree iff it has tree width 1.
Hypergraphs with tree widths larger than 1 are not trees; for example, the family of cliques has unbounded tree width.

The definition relies on the concept of a \emph{tree decomposition}.
For Robertson and Seymour~\cite{robertson1986graph-minorsII}, a decomposition is itself a tree \(Y\), each vertex of which is associated with a subhypergraph of \(G\).
Then \(G\) can be reconstructed from  \(Y\) by identifying some vertices.

\begin{defiC}[\cite{robertson1986graph-minorsII}]\label{def:tree-decomposition}
  A \emph{tree decomposition} of a hypergraph \(G = \mathgraph{E}{V}\) is a pair \((Y,t)\) where \(Y\) is a tree and \(t \colon \vertices(Y) \to \parti(V)\) is a function such that:
  \begin{enumerate}
    \item Every vertex is in one of the components:
      \(\Union_{i \in \vertices(Y)} t(i) = V\).
    \item Every edge has its endpoints in a component:
      \(\forall e \in E \ \exists i \in \vertices(Y) \ \edgeends{e} \subseteq t(i)\).
    \item The components are glued in a tree shape:
      \(\forall i,j,k \in \vertices(Y) \ i \graphpath j \graphpath k \implies t(i) \intersection t(k) \subseteq t(j)\).
  \end{enumerate}
\end{defiC}

The cost is the maximum number of vertices of the component subhypergraphs.

\begin{exa}
  Consider the hypergraph $G$ and its tree decomposition \((Y,t)\) below.
  Its cost is \(3\) as its biggest component has three vertices.
  \[\treeDecExFig{}\]
\end{exa}

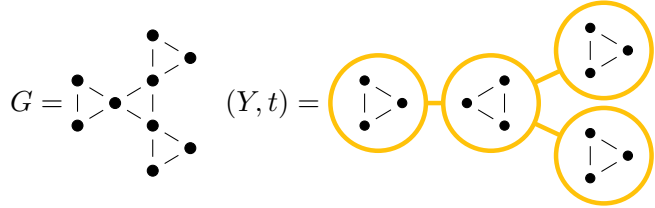
\begin{defi}[Tree width]
  Given a tree decomposition \((Y,t)\) of a hypergraph \(G\), its width is \(\decwidth(Y,t) \defn \max_{i \in \vertices(Y)} \card{t(i)}\).
  The tree width of \(G\) is given by the min-max formula:
  \[\treewidth(G) \defn \min_{(Y,t)} \decwidth(Y,t).\]
\end{defi}

Note that Robertson and Seymour subtract \(1\) from \(\treewidth(G)\) so that trees have tree width \(1\). To minimise bureaucratic overhead, we ignore this convention.\\

We use branch width~\cite{robertson1991graph-minorsX} as a technical stepping stone to relate monoidal width and tree width. Before presenting its definition, it is important to note that branch width and tree width are \emph{equivalent}, i.e.\ they are within a constant factor of each other.

\begin{thmC}[{\cite[Theorem 5.1]{robertson1991graph-minorsX}}]\label{th:twd-bwd}
  Branch width is equivalent to tree width.
  More precisely, for a hypergraph \(G = \mathgraph{E}{V}\),
  \[\max\{\branchwidth(G),\gamma(G)\} \leq \treewidth(G) \leq \max \{ \dfrac{3}{2} \branchwidth(G), \gamma(G), 1\}.\]
\end{thmC}

Branch width relies on branch decompositions, which, intuitively, record in a tree a way of iteratively partitioning the edges of a hypergraph.

\begin{defiC}[\cite{robertson1991graph-minorsX}]\label{def:branch-decomposition}
  A \emph{branch decomposition} of a hypergraph \(G = \mathgraph{E}{V}\) is a pair \((Y,b)\) where \(Y\) is a subcubic tree and \(b \colon \leaves(Y) \cong E\) is a bijection.
\end{defiC}

Each edge \(e\) in the tree \(Y\) determines a splitting of the hypergraph.
More precisely, it determines a two partition of the leaves of \(Y\), which, through \(b\), determines a 2-partition \(\{A_e,B_e\}\) of the edges of \(G\).
This corresponds to a splitting of the hypergraph \(G\) into two subhypergraphs \(G_{1}\) and \(G_{2}\).
Intuitively, the order of an edge \(e\) is the number of vertices that are glued together when joining \(G_{1}\) and \(G_{2}\) to get \(G\).
Given the partition \(\{A_e,B_e\}\) of the edges of \(G\), we say that a vertex \(v\) of \(G\) separates \(A_e\) and \(B_e\) whenever there are an edge in \(A_e\) and an edge in \(B_e\) that are both adjacent to \(v\).

Let \((Y,b)\) be a branch decomposition of a hypergraph \(G\).
Let \(e\) be an edge of \(Y\).
The \emph{order} of \(e\) is the number of vertices that separate \(A_e\) and \(B_e\): \(\edgeorder(e) \defn \card{\edgesetends{A_e} \intersection \edgesetends{B_e}}\).

\begin{defi}[Branch width]\label{def:branch-width}
  Given a branch decomposition \((Y,b)\) of a hypergraph \(G=\mathgraph{E}{V}\), define its width as \(\decwidth(Y,b) \defn \max_{e \in \edges(Y)} \edgeorder(e)\).

  The branch width of \(G\) is given by the min-max formula:
  \(\branchwidth(G) \defn \min_{(Y,b)} \decwidth(Y,b)\).
\end{defi}

\begin{exa}
  If we start reading the decomposition from an edge in the tree \(Y\), we can extend the labelling to internal vertices by labelling them with the glueing of the labels of their children.
  \[\branchDecExFig{}\]
  In this example, there is only one vertex separating the first two subgraphs of the decomposition.
  This means that the corresponding edge in the decomposition tree has order \(1\).
\end{exa}

\subsection{Hypergraphs with sources and inductive definition}\label{sec:rec-branch-dec}
We introduce a definition of decomposition that is intermediate between a branch decomposition and a monoidal decomposition.
It adds to branch decompositions the algebraic flavour of monoidal decompositions by using an inductive data type, that of binary trees, to encode a decomposition.

Our approach follows closely Bauderon and Courcelle's hypergraphs with sources~\cite{bauderon1987graph} and the corresponding inductive definition of tree decompositions~\cite{courcelle1992monadic}.
Courcelle's result~\cite[Theorem 2.2]{courcelle1992monadic} is technically involved as it translates between a combinatorial description of a decomposition to a syntactic one.
Our results in this and the next sections are similarly technically involved.

We recall the definition of hypergraphs with sources and introduce inductive branch decompositions of them.
Intuitively, the sources of a hypergraph are marked vertices that are allowed to be \virgolette{glued} together with the sources of another hypergraph.
Thus, the equivalence between branch decompositions and inductive branch decompositions formalises the intuition that a branch decomposition encodes a way of dividing a hypergraph into smaller subgraphs by \virgolette{cutting} along some vertices.

\begin{defiC}[\cite{bauderon1987graph}]
  A \emph{hypergraph with sources} is a pair \(\Gamma = (G,X)\) where \(G = \mathgraph{E}{V}\) is a hypergraph and \(X \subseteq V\) is a subset of its vertices, called the sources (\Cref{ex:graph-with-sources}).
  Given two hypergraphs with sources \(\Gamma = (G,X)\) and \(\Gamma' = (G',X')\), we say that \(\Gamma'\) is a subhypergraph of \(\Gamma\) whenever \(G'\) is a subhypergraph of \(G\).
\end{defiC}
Note that the sources of a subhypergraph \(\Gamma'\) of \(\Gamma\) need not to appear as sources of \(\Gamma\), nor vice versa.
In fact, if \(\Gamma\) is obtained by identifying all the sources of \(\Gamma_{1}\) with some of the sources of \(\Gamma_{2}\), the sources of \(\Gamma\) and \(\Gamma_{1}\) will be disjoint.
\begin{figure}[h!]
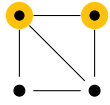

  \[\graphSourcesExFig{}\]
  \caption{Sources are marked vertices in the graph and are thought of as an interface that can be glued with that of another graph.}\label{ex:graph-with-sources}
\end{figure}

An inductive branch decomposition is a binary tree whose vertices carry subhypergraphs \(\Gamma'\) of the ambient hypergraph \(\Gamma\).  This set of all such binary trees is defined as follows
 \[\binarytrees{\Gamma} \ \Coloneqq \ \emptydec \ \mid \ (\binarytrees{\Gamma}, \Gamma', \binarytrees{\Gamma}) \]
where \(\Gamma'\) ranges over the non-empty subhypergraphs of \(\Gamma\).
An inductive branch decomposition has to satisfy additional conditions that ensure that ``glueing'' \(\Gamma_1\) and \(\Gamma_2\) together yields \(\Gamma\).

\begin{defi}\label{def:rec-branch-decomposition}
  Let \(\Gamma = (\mathgraph{E}{V},X)\) be a hypergraph with sources.
  An \emph{inductive branch decomposition} of \(\Gamma\) is \(T \in \binarytrees{\Gamma}\) where either:
  \begin{itemize}
  \item \(\Gamma\) is discrete (i.e. it has no edges) and \(T = \emptydec\);
  \item \(\Gamma\) has one edge and \(T = \nodegenerator{\emptydec}{\Gamma}{\emptydec}\). We will use the shorthand \(T = \leafgenerator{\Gamma}\) in this case;
  \item \(T = \nodegenerator{T_1}{\Gamma}{T_2}\) and \(T_i \in T_{\Gamma_{i}}\) are inductive branch decompositions of subhypergraphs \(\Gamma_i = (\mathgraph{E_i}{V_i},X_i)\) of \(\Gamma\) such that:
  \begin{itemize}
    \item The edges are partitioned in two, \(E = E_1 \disjointunion E_2\) and \(V = V_1 \union V_{2}\);
    \item The sources are those vertices shared with the original sources as well as those shared with the other subhypergraph, \(X_i = (V_1 \intersection V_2) \union (X \intersection V_i)\).
  \end{itemize}
\end{itemize}
\end{defi}

Note that \(\edgeends{E_i} \subseteq V_i\) and that not all subtrees of a decomposition \(T\) are themselves decompositions: only those \(T'\) that contain all the nodes in \(T\) that are below the root of \(T'\).
We call these \emph{full} subtrees and indicate with \(\labelling(T')\) the subhypergraph of \(\Gamma\) that \(T'\) is a decomposition of.
We sometimes write \(\Gamma_i = \labelling(T_i)\), \(V_i = \vertices(\Gamma_i)\) and \(X_i = \boundary(\Gamma_i)\).
Then,
\begin{equation}\label{eq:expr-boundary}
  \boundary(\Gamma_i) = (\vertices(\Gamma_1) \intersection \vertices(\Gamma_2)) \union (\boundary(\Gamma) \intersection \vertices(\Gamma_i)).
\end{equation}

\begin{defi}
  Let \(T = \nodegenerator{T_1}{\Gamma}{T_2}\) be an inductive branch decomposition of \(\Gamma = (G,X)\), with \(T_i\) possibly both empty.
  Define the \emph{width} of \(T\) inductively:
  \(\decwidth(\emptydec) \defn 0\),
  and \(\decwidth(T) \defn \max\{\decwidth(T_1), \decwidth(T_2),\card{\boundary(\Gamma)}\}\).
  Expanding this expression, we obtain
  \[\decwidth(T) = \max_{T' \text{ full subtree of } T} \card{\boundary(\labelling(T'))}.\]
  The \emph{inductive branch width} of \(\Gamma\) is defined by the min-max formula  \(\rbranchwidth(\Gamma) \defn \min_{T} \decwidth(T)\).
\end{defi}

We show that this definition is equivalent to the original one by exhibiting a mapping from branch decompositions to inductive branch decompositions that preserve the width and vice versa.
Showing that these mappings preserve the width is a bit involved because the order of the edges in a decomposition is defined ``globally'', while, for an inductive decomposition, the width is defined inductively.
Thus, we first need to show that we can compute the inductive width globally.

\begin{lem}\label{lemma:expr-boundaries}
  Let \(\Gamma = (G,X)\) be a hypergraph with sources and \(T\) be an inductive branch decomposition of  \(\Gamma\).
  Let \(T_{0}\) be a full subtree of \(T\) and let \(T'\ngtrless T_{0}\) denote a full subtree \(T'\) of \(T\) such that its intersection with \(T_{0}\) is empty.
  Then,
  \[\boundary(\labelling(T_{0})) = \vertices(\labelling(T_{0})) \intersection \left(X \union \Union_{T' \ngtrless T_{0}} \vertices(\labelling(T'))\right).\]
\end{lem}
\begin{proof}
  Proceed by induction on the decomposition tree \(T\).
  If it is empty, \(T = \emptydec\), then its subtree is also empty, \(T_0 = \emptydec\), and we are done.

  If \(T = \nodegenerator{T_1}{\Gamma}{T_2}\), then either \(T_{0}\) is a full subtree of \(T_{1}\), or it is a full subtree of \(T_{2}\), or it coincides with \(T\).
  If \(T_0\) coincides with \(T\), then their boundaries coincide and the statement is satisfied because \(\boundary(\labelling(T_0)) = X = V \intersection X\).
  Now suppose that \(T_0\) is a full subtree of \(T_{1}\).
  Then, by applying the induction hypothesis, \Cref{eq:expr-boundary}, and using the fact that \(\labelling(T_0) \subgrapheq \labelling(T_1)\), we compute the sources of \(T_{0}\):
  \begin{align*}
    & \boundary(\labelling(T_0)) \\
    & = \vertices(\labelling(T_0)) \intersection \left(\boundary(\labelling(T_1)) \union \Union_{T' \subtreeq T_1, T' \ngtrless T_0} \vertices(\labelling(T'))\right)\\
    & = \vertices(\labelling(T_0)) \intersection \left(\left(\vertices(\labelling(T_1)) \intersection (\vertices(\labelling(T_2)) \union X)\right) \union \Union_{T' \subtreeq T_1, T' \ngtrless T_0} \vertices(\labelling(T'))\right)\\
    & = \vertices(\labelling(T_0)) \intersection \left(\vertices(\labelling(T_2)) \union X \union \Union_{T' \subtreeq T_1, T' \ngtrless T_0} \vertices(\labelling(T'))\right)\\
    & = \vertices(\labelling(T_0)) \intersection \left(X \union \Union_{T' \subtreeq T, T' \ngtrless T_0} \vertices(\labelling(T'))\right)
  \end{align*}
  A similar computation can be done if \(T_0\) is a full subtree of \(T_{2}\).
\end{proof}

\begin{lem}\label{lemma:rec-bwd-upper}
  Let \(\Gamma = (G,X)\) be a hypergraph with sources and \(G = \mathgraph{E}{V}\) be its underlying hypergraph.
  Let \(T\) be an inductive branch decomposition of \(\Gamma\).
  Then, there is a branch decomposition \(\fromRecursiveDec(T)\) of \(G\) such that \(\decwidth(\fromRecursiveDec(T)) \leq \decwidth(T)\).
\end{lem}
\begin{proof}
  A binary tree is, in particular, a subcubic tree.
  Then, we can define \(Y\) to be the unlabelled tree underlying \(T\).
  The label of a leaf \(l\) of \(T\) is a subhypergraph of \(\Gamma\) with one edge \(e_l\).
  Then, there is a bijection \(b \colon \leaves(T) \to \edges(G)\) such that \(b(l) \defn e_l\).
  Then, \((Y,b)\) is a branch decomposition of \(G\) and we can define \(\fromRecursiveDec(T) \defn (Y,b)\).

  By construction, \(e \in \edges(Y)\) if and only if \(e \in \edges(T)\).
  Let \(\{v,w\} = \edgeends{e}\) with \(v\) parent of \(w\) in \(T\) and let \(T_w\) the full subtree of \(T\) with root \(w\).
  Let \(\{E_v,E_w\}\) be the (non-trivial) partition of \(E\) induced by \(e\).
  Then, for the edges sets, \(E_w = \edges(\labelling(T_{w}))\) and \(E_v = \Union_{T' \ngtrless T_{w}} \edges(\labelling(T'))\), and, for the vertices sets, \(\edgesetends{E_w} \subseteq \vertices(\labelling(T_{w}))\) and \(\edgesetends{E_v} \subseteq \Union_{T' \ngtrless T_{w}} \vertices(\labelling(T'))\).
  Using these inclusions and applying \Cref{lemma:expr-boundaries},
  \begin{align*}
    & \edgeorder(e) && \decwidth(Y,b)\\
    & \defn \card{\edgesetends{E_w} \intersection \edgesetends{E_v}} && \defn \max_{e \in \edges(Y)} \edgeorder(e)\\
    & \leq \card{\vertices(\labelling(T_{w})) \intersection \Union_{T' \ngtrless T_{w}} \vertices(\labelling(T'))} && \leq \max_{T' < T} \card{\boundary(\labelling(T'))}\\
    & \leq \card{\vertices(\labelling(T_{w})) \intersection (X \union \Union_{T' \ngtrless T_{w}} \vertices(\labelling(T')))} && \leq \max_{T' \leq T} \card{\boundary(\labelling(T'))}\\
    & = \card{\boundary(\labelling(T_{w}))} && = \decwidth(T)
    \qedhere
  \end{align*}
\end{proof}

\begin{lem}\label{lemma:rec-bwd-lower}
  Let \(\Gamma = (G,X)\) be a hypergraph with sources and \(G = \mathgraph{E}{V}\) be its underlying hypergraph.
  Let \((Y,b)\) be a branch decomposition of \(G\).
  Then, there is a branch decomposition \(\toRecursiveDec(Y,b)\) of \(\Gamma\) such that \(\decwidth(\toRecursiveDec(Y,b)) \leq \decwidth(Y,b) + \card{X}\).
\end{lem}
\begin{proof}
  Proceed by induction on \(\card{\edges(Y)}\).
  If \(Y\) has no edges, then either \(G\) has no edges and \((Y,b) = \emptydec\) or \(G\) has only one edge \(e_l\) and \((Y,b) = \leafgenerator{e_l}\).
  In either case, define \(\toRecursiveDec(Y,b) \defn \leafgenerator{\Gamma}\) and \(\decwidth(\toRecursiveDec(Y,b)) \defn \card{X} \leq \decwidth(Y,b) + \card{X}\).

  If \(Y\) has at least one edge \(e\), then \(Y = Y_1 \overset{e}{\text{---}} Y_2\) with \(Y_i\) a subcubic tree.
  Let \(E_{i} = b(\leaves(Y_{i}))\) be the sets of edges of \(G\) indicated by the leaves of \(Y_{i}\).
  Then, \(E_{1} \disjointunion E_{2} = E\).
  By induction hypothesis, there are inductive branch decompositions \(T_{i} \defn \toRecursiveDec(Y_{i},b_{i})\) of \(\Gamma_{i} = (G_{i}, X_{i})\), where \(V_{1} \defn \edgesetends{E_{1}}\), \(V_{2} \defn \edgesetends{E_{2}} \union (V \setminus V_{1})\), \(X_{i} \defn (V_{1} \intersection V_{2}) \union (V_{i} \intersection X)\) and \(G_{i} \defn \mathgraph{E_{i}}{V_{i}}\).
  Then, the tree \(\toRecursiveDec(Y,b) \defn \nodegenerator{T_1}{\Gamma}{T_2}\) is an inductive branch decomposition of \(\Gamma\) and, by applying \Cref{lemma:expr-boundaries},
  \begin{align*}
    & \decwidth(\toRecursiveDec(Y,b)) \\
    & \defn \max \{\decwidth(T_{1}), \card{X}, \decwidth(T_{2})\}\\
    & = \max_{T' \leq T} \card{\boundary(\labelling(T'))}\\
    & \leq \max_{T' \leq T} \card{\vertices(\labelling(T')) \intersection \edgesetends{E \setminus \edges(\labelling(T'))}} + \card{X} \\
    & = \max_{e \in \edges(Y)} \edgeorder(e) + \card{X}\\
    & \codefn \decwidth(Y,b) + \card{X}
    \qedhere
  \end{align*}
\end{proof}

Combining \Cref{lemma:rec-bwd-upper} and \Cref{lemma:rec-bwd-lower} we obtain:

\begin{prop}\label{prop:rec-branch-width-equivalent}
  Inductive branch width is equivalent to branch width.
\end{prop}

\subsection{Cospans of hypergraphs}\label{sec:cospans-hypergraphs}
We work with the category \(\UHGraph\) of undirected hypergraphs and their homomorphisms (\Cref{def:hypergraph}).
The monoidal category \(\cospanUHGraph\) of cospans is a standard choice for an algebra of \virgolette{open} hypergraphs. Hypergraphs are composed by glueing vertices~\cite{rosebrugh2005cospangraph,gadducci1997inductive,fong2015cospans}.
We do not need the full expressivity of \(\cospanUHGraph\) and restrict to \(\cospanUHGraphO\), where the objects are sets, seen as discrete hypergraphs.

\begin{defi}
  A \emph{cospan} in a category \(\cat{C}\) is a pair of morphisms in \(\cat{C}\) that share the same codomain, called the \emph{head}, \(f \colon X \to E\) and \(g \colon Y \to E\).
  When \(\cat{C}\) has finite colimits, cospans form a symmetric monoidal category \(\catCospan{\cat{C}}\) whose objects are the objects of \(\cat{C}\) and morphisms are cospans in \(\cat{C}\).
  More precisely, a morphism \(X \to Y\) in \(\catCospan{\cat{C}}\) is an equivalence class of cospans \(\cospan{X}{f}{E}{g}{Y}\), up to isomorphism of the head of the cospan.
  The composition of \(\cospan{X}{f}{E}{g}{Y}\) and \(\cospan{Y}{h}{F}{l}{Z}\) is given by the pushout of \(g\) and \(h\).
  Intuitively, the pushout of \(g\) and \(h\) ``glues'' \(E\) and \(F\) along the images of \(g\) and \(h\) (see \Cref{ex:cospan-composition}).
  The monoidal product is given by component-wise coproducts.
\end{defi}

We can construct the category of cospans of hypergraphs \(\cospanUHGraph\) because the category of hypergraphs \(\UHGraph\) has all finite colimits.

\begin{prop}
  The category \(\UHGraph\) has all finite colimits and they are computed pointwise.
\end{prop}
\begin{proof}
  Let \(\fun{D} \colon \cat{J} \to \UHGraph\) be a diagram in \(\UHGraph\).
  Then, every object \(i\) in \(\cat{J}\) determines a hypergraph \(G_i \defn \fun{D}(i) = \mathgraph{E_i,\edgeendsfun_i}{V_i}\) and every \(f \colon i \to j\) in \(\cat{J}\), gives a hypergraph homomorphism \(\fun{D}(f)=(f_V,f_E)\).
  Let the functors \(\forgetfulFun_{E} \colon \UHGraph \to \Set\) and \(\forgetfulFun_{V} \colon \UHGraph \to \Set\) associate the edges, resp. vertices, component to hypergraphs and hypergraph homomorphisms: for a hypergraph \(G = \mathgraph{E}{V}\), \(\forgetfulFun_E(G) \defn E\) and \(\forgetfulFun_V(G) \defn V\); and, for a morphism \(f = (f_V,f_E)\), \(\forgetfulFun_E(f) \defn f_E\) and \(\forgetfulFun_V(f) \defn f_V\).
  \[\diagramforgetfulsDiagram{}\]
  The category \(\Set\) has all colimits, thus there are \(E_0 \defn \colim(\fun{D} \dcomp \forgetfulFun_{E})\) and \(V_0 \defn \colim(\fun{D} \dcomp \forgetfulFun_{V})\).
  Let \(c_{i} \colon V_i \to V_0\) and \(d_{i} \colon E_i \to E_0\) be the inclusions given by the colimits.
  Then, for any \(i,j \in \obj{\cat{J}}\) the following diagrams commute:
  \[\vertexcolimitcomponentsDiagram{} \qquad \edgecolimitcomponentsDiagram{}\]
  By definition of hypergraph morphism, \(f_{E} \dcomp \edgeendsfun_{j} = \edgeendsfun_{i} \dcomp \parti(f_{V})\), and, by functoriality of \(\parti\), \(\parti(f_{V}) \dcomp \parti(c_{j}) = \parti(c_{i})\).
  This shows that \(\parti(V_0)\) is a cocone over \(\fun{D} \dcomp \forgetfulFun_{E}\) with morphisms given by \(\edgeendsfun_{i} \dcomp \parti(c_{i})\).
  Then, there is a unique morphism \(\edgeendsfun \colon E_0 \to \parti(V_0)\) that commutes with the cocone morphisms: \(d_i \dcomp \edgeendsfun = \edgeendsfun_i \dcomp \parti(c_i)\).
  \[\edgeendscolimitmorphismDiagram{}\]
  This shows that the pairs \((c_{i},d_{i})\) are hypergraph morphisms and, with the hypergraph defined by \(G_0 \defn \mathgraph{E_0, \edgeendsfun}{V_0}\), form a cocone over \(\fun{D}\) in \(\UHGraph\).
  Let \(H = \mathgraph{E_H, \edgeendsfun_H}{V_H}\) be another cocone over \(\fun{D}\) with morphisms \((a_i,b_i) \colon G_i \to H\).
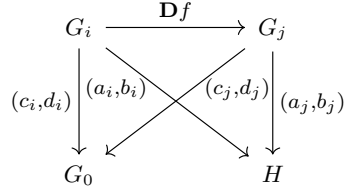
  \[\graphcolimitcomponentsDiagram{}\]
  We show that \(G_0\) is initial by constructing a morphism \((h_V,h_E) \colon G_0 \to H\) and showing that it is the unique one commuting with the inclusions.

  By applying the functors \(\forgetfulFun_E\) and \(\forgetfulFun_V\) to the diagram above, we obtain the following diagrams in \(\Set\), where \(h_V \colon V_0 \to V_H\) and \(h_E \colon E_0 \to E_H\) are the unique morphism from the colimit cone.
  \[\vertexcoconecomponentsDiagram{} \qquad \edgecoconecomponentsDiagram{}\]

  We show that \((h_V, h_E)\) is a hypergraph morphism.
  The object \(\parti(V_H)\) is a cocone over \(\fun{D} \dcomp \forgetfulFun_E\) in (at least) two ways: with morphisms \(d_i \dcomp \edgeendsfun \dcomp \parti(h_V)\) and morphisms \(b_i \dcomp \edgeendsfun_H\).
  By initiality of \(E_0\), there is a unique morphism \(E_0 \to \parti(V_H)\) and it must coincide with \(h_E \dcomp \edgeendsfun_H\) and \(\edgeendsfun \dcomp \parti(h_V)\).
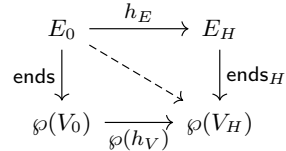
  \[\universalarrowgraphmorphisDiagram{}\]
  This proves that \((h_V,h_E)\) is a hypergraph morphism.
  It is, moreover, unique because any other morphism with this property would have the same components.
  In fact, let \((h'_V,h'_E) \colon G_0 \to H\) be another hypergraph morphism that commutes with the cocones, i.e. \((c_i,d_i) \dcomp (h'_V,h'_E) = (a_i,b_i)\).
  Then, its components must commute with the respective cocones in \(\Set\), by functoriality of \(\forgetfulFun_E\) and \(\forgetfulFun_V\): \(c_i \dcomp h'_V = a_i\) and \(d_i \dcomp h'_E = b_i\).
  By construction, \(V_0\) and \(E_0\) are the colimits of \(\fun{D} \dcomp \forgetfulFun_V\) and \(\fun{D} \dcomp \forgetfulFun_E\), so there are unique morphisms to any other cocone over the same diagrams.
  This means that \(h'_V = h_V\) and \(h'_E = h_E\), which shows the uniqueness of \((h_V,h_E)\).
\end{proof}

\begin{defi}
  The category \(\cospanUHGraphO\) is the full subcategory of\linebreak[4]\(\cospanUHGraph\) on discrete hypergraphs.
  Objects are sets and a morphism \(g \colon X \to Y\) is given by a hypergraph \(G = \mathgraph{E}{V}\) and two functions, \(\partial_{X} \colon X \to V\) and \(\partial_{Y} \colon Y \to V\).
\end{defi}

Composition in \(\cospanUHGraphO\) is given by identification of the common sources: if two vertices are pointed by a common source, then they are identified.

\begin{exa}\label{ex:cospan-composition}
  The composition of two morphisms with a single edge along a common vertex gives a path of length two, obtained by identifying the vertex \(v\) of the first morphism with the vertex \(u\) of the second.
  \[\cospancompositionexampleFig{}\]
\end{exa}

\subsection{String diagrams for cospans of hypergraphs}\label{sec:frobunnius-graphs}
We introduce a syntax for the monoidal category \(\cospanUHGraphO\), which we will use for proving some of the results in this section.
We will show that the syntax for \(\cospanUHGraphO\) is given by the syntax of \(\cospanSet\) together with an extra \virgolette{hyperedge} generator \(\oneedge_{n} \colon n \to 0\) for every \(n \in \naturals\).
This result is inspired by the similar one for cospans of directed graphs~\cite{rosebrugh2005cospangraph}.

It is well-known that the category \(\cospanSet\) of finite sets and cospans of functions between them has a convenient syntax given by the walking special Frobenius monoid~\cite{lack04composing-props}.

\begin{propC}[\cite{lack04composing-props}]\label{prop:syntax-cospan-set}
  The skeleton of the monoidal category \(\cospanSet\) is isomorphic to the prop \(\sFrob\), whose generators and axioms are in \Cref{fig:frobunnius}.
\end{propC}

\begin{figure}
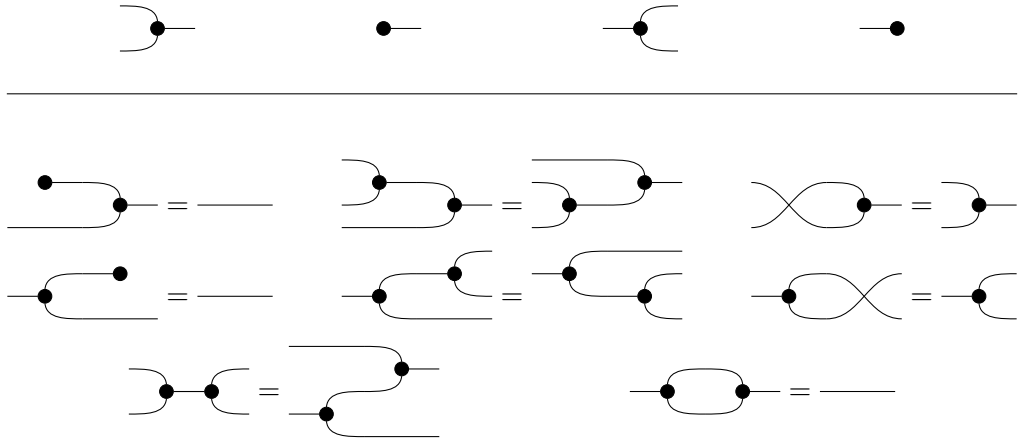

  \specialFrobunniusFig{}
  \caption{Generators and axioms of a special Frobenius monoid.}\label{fig:frobunnius}
\end{figure}

In order to obtain cospans of hypergraphs from cospans of sets, we need to add generators that behave like hyperedges: they have \(n\) inputs and these inputs can be permuted without any effect.

\begin{defi}
  Define \(\UHedge\) to be the prop generated by a ``hyperedge'' generator \(\oneedge_n \colon n \to 0\) for every \(n \in \naturals\) such that permuting its inputs does not have any effect:
  \[\forall n \in \naturals \quad \hyperedgenFig \quad \text{such that} \quad \forall \text{ permutation } \sigma \colon n \to n  \quad\uhedgepermutationFig = \hyperedgenFig\]
\end{defi}

The syntax for cospans of graphs is defined as a coproduct of props.

\begin{defi}
  Define the prop \(\sFrobUHGraph\) as a coproduct: \(\sFrobUHGraph \defn \sFrob + \UHedge\).
\end{defi}

We will show that every morphism \(g \colon n \to m\) in \(\sFrobUHGraph\) corresponds to a morphism in \(\cospanUHGraphO\).

\begin{exa}
  The string diagram below corresponds to a hypergraph with two left sources, one right source and two hyperedges.
  The number of endpoints of each hyperedge is given by the arity of the corresponding generator in the string diagram.
  Two hyperedges are adjacent to the same vertex when they are connected by the Frobenius structure in the string diagram, and a hyperedge is adjacent to a source when it is connected to an input or output in the string diagram.
  \[\cospangraphisoexampleFig{}\]
\end{exa}

\begin{prop}
  There is a symmetric monoidal functor \(\fun{S} \colon \sFrobUHGraph \to \cospanUHGraphO\).
\end{prop}
\begin{proof}
  By definition, \(\sFrobUHGraph \defn \sFrob + \UHedge\) is a coproduct.
  Therefore, it suffices to define two symmetric monoidal functors \(\fun{S_1} \colon \sFrob \to \cospanUHGraphO\) and \(\fun{S_2} \colon \UHedge \to \cospanUHGraphO\) for constructing the functor \(\fun{S} \defn \coproductmap{\fun{S_1}}{\fun{S_2}}\).

  The category of cospans of finite sets embeds into the category of cospans of undirected hypergraphs, and in particular \(\cospanSet \into \cospanUHGraphO\).
  By \Cref{prop:syntax-cospan-set}, there is a functor \(\sFrob \to \cospanSet\), which gives us a functor \(\fun{S_1} \colon \sFrob \to \cospanUHGraphO\).

  For the functor \(\fun{S_2}\), we need to define it on the generators of \(\UHedge\) and show that it preserves the equations.
  We define \(\fun{S_2}(\oneedge_n)\) to be the cospan of graphs \(\cospan{n}{}{\mathgraph{\{e\}}{n}}{}{\emptyset}\) given by \(\id{n} \colon n \to n\) and \(\initmap{n} \colon \emptyset \to n\).
  With this assignment, we can freely extend \(\fun{S_2}\) to a monoidal functor \(\UHedge \to \cospanUHGraphO\).
  In fact, it preserves the equations of \(\UHedge\) because permuting the order of the endpoints of an undirected hyperedge has no effect by definition.
\end{proof}

In order to instantiate monoidal width in \(\cospanUHGraphO\), we need to define an appropriate weight function.
\begin{defi}\label{def:weight-cospan-graph}
  Let \(\decgenerators\) be all morphisms of \(\cospanUHGraphO\).
  Define the \emph{weight function} as follows.
  For an object \(X\), \(\nodeweight(\dcompnode{X}) \defn \card{X}\).
  For a morphism \(g \in \decgenerators\), \(\nodeweight(g) \defn \card{V}\), where \(V\) is the set of vertices of the apex of \(g\), i.e. \(g = \cospan{X}{}{G}{}{Y}\) and \(G = \mathgraph{E}{V}\).
\end{defi}

\subsection{Tree width as monoidal width}\label{sec:mwd-twd-proof}
Here we show that monoidal width in the monoidal category \(\cospanUHGraphO\), with the weight function given in \Cref{def:weight-cospan-graph}, is equivalent to tree width.
We do this by bounding monoidal width by above with branch width \(+ 1\) and by below with half of branch width (\Cref{th:mwd-bwd}).
We prove these bounds by defining maps from inductive branch decompositions to monoidal decompositions that preserve the width (\Cref{prop:mwd-bwd-upper}), and vice versa (\Cref{prop:mwd-bwd-lower}).

The idea behind the mapping from inductive branch decompositions to monoidal decompositions is to take a one-edge hypergraph for each leaf of the inductive branch decomposition and compose them following the structure of the decomposition tree.
The \(3\)-clique has a branch decomposition as shown on the left.
The corresponding monoidal decomposition is shown on the right.
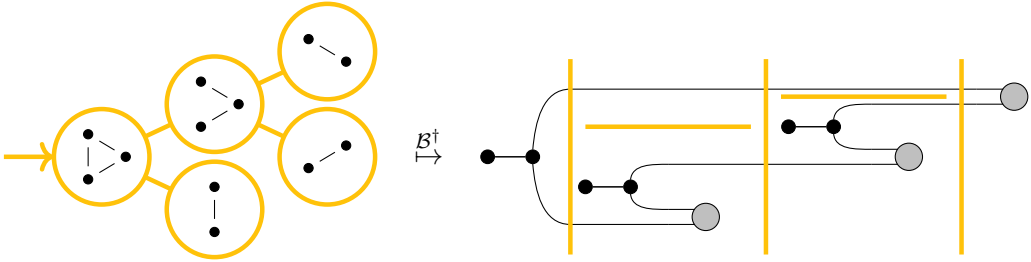
\[\bwdMappingExFig{}\]

\begin{prop}\label{prop:mwd-bwd-upper}
  Let \(\Gamma = (G,X)\) be a hypergraph with sources and \(T\) be an inductive branch decomposition of \(\Gamma\).
  Let \(g \defn \cospan{X}{\inclusion}{G}{}{\emptyset}\) be the corresponding cospan.
  Then, there is a monoidal decomposition \(\bTomdec(T) \in \decset{g}\) such that \(\decwidth(\bTomdec(T)) \leq \max\{\decwidth(T) + 1, \hyperedgesize(G)\}\).
\end{prop}
\begin{proof}
  Let \(G = \mathgraph{E}{V}\) and proceed by induction on the decomposition tree \(T\).
  If the tree \(T = \leafgenerator{\Gamma}\) is composed of only a leaf, then the label \(\Gamma\) of this leaf must have only one hyperedge with \(\hyperedgesize(G)\) endpoints and \(\decwidth(T) \defn \card{X}\).
  We define the corresponding monoidal decomposition to also consist of only a leaf, \(\bTomdec(T) \defn (g)\), and obtain the desired bound \(\decwidth(\bTomdec(T)) = \max\{\card{X}, \hyperedgesize(G)\} = \max \{\decwidth(T), \hyperedgesize(G)\}\).

  If \(T = \nodegenerator{T_1}{\Gamma}{T_2}\), then, by definition of branch decomposition, \(T\) is composed of two subtrees \(T_1\) and \(T_{2}\) that give branch decompositions of \(\Gamma_1 = (G_1,X_1)\) and \(\Gamma_2 = (G_2,X_2)\).
  There are three conditions imposed by the definition on these subgraphs \(G_i = \mathgraph{E_i}{V_i}\): \(E = E_1 \disjointunion E_2\) with \(E_i \neq \emptyset\), \(V_1 \union V_2 = V\), and \(X_i = (V_1 \intersection V_2) \union (X \intersection V_i)\).
  Let \(g_i = \cospan{X_i}{}{G_i}{}{\emptyset}\) be the cospan given by \(\inclusion \colon X_{i} \to V_{i}\) and corresponding to \(\Gamma_{i}\).
  Then, we can decompose \(g\) in terms of identities, the structure of \(\cospanUHGraphO\), and its subgraphs \(g_{1}\) and \(g_{2}\):
  \[\bwdProofFig{}\]
  By induction hypothesis, there are monoidal decompositions \(\bTomdec(T_i)\) of \(g_i\) whose width is bounded: \(\decwidth(\bTomdec(T_i)) \leq \max\{\decwidth(T_i) + 1, \hyperedgesize(G_i)\}\).
  By \Cref{lemma:mwd-copy}, there is a monoidal decomposition \(\copyMdec(\bTomdec(T_1))\) of the morphism in the above dashed box of bounded width: \(\decwidth(\copyMdec(\bTomdec(T_1))) \leq \max\{\decwidth(\bTomdec(T_1)), \card{X_1}+1\}\).
  Using this decomposition, we can define the monoidal decomposition given by the cuts in the figure above.
  \[\bTomdec(T) \defn \nodegenerator{\nodegenerator{\copyMdec(\bTomdec(T_1))}{\tensor}{\id{X_{2}\setminus X_{1}}}}{\dcomp_{X_2}}{\bTomdec(T_2)}.\]
  We can bound its width by applying \Cref{lemma:mwd-copy}, the induction hypothesis and the relevant definitions of width (\Cref{def:branch-width} and \Cref{def:weight-cospan-graph}).
  \begin{align*}
    & \decwidth(\bTomdec(T))\\
    & \defn \max \{\decwidth(\copyMdec(\bTomdec(T_1))), \decwidth(\bTomdec(T_2)), \card{X_2}\}\\
    & = \max \{\decwidth(\bTomdec(T_1)), \decwidth(\bTomdec(T_2)), \card{X_1} + 1, \card{X_2}\}\\
    & \leq \max \{\decwidth(T_1) + 1, \hyperedgesize(G_1), \decwidth(T_2) + 1, \hyperedgesize(G_2), \card{X_1} + 1, \card{X_2}\}\\
    & \leq \max \{\max\{\decwidth(T_1), \decwidth(T_2), \card{X_1}, \card{X_2}\} + 1, \hyperedgesize(G_1), \hyperedgesize(G_2)\}\\
    & \leq \max \{\max\{\decwidth(T_1), \decwidth(T_2), \card{X}\} + 1, \hyperedgesize(G)\}\\
    & \codefn \max\{\decwidth(T) +1, \hyperedgesize(G)\}
  \qedhere
  \end{align*}
\end{proof}

The mapping \(\mTobdec\) follows the same idea of the mapping \(\bTomdec\) but requires extra care: we need to keep track of which vertices are going to be identified in the final cospan.
The function \(\phi\) stores this information, thus it cannot identify two vertices that are not already in the boundary of the hypergraph.
The proof of \Cref{prop:mwd-bwd-lower} proceeds by induction on the monoidal decomposition and constructs the corresponding branch decomposition.
The inductive step relies on \(\phi\) to identify which subgraphs of \(\Gamma\) correspond to the two subtrees in the monoidal decomposition, and, consequently, to define the corresponding branch decomposition.

\begin{rem}\label{rem:compute-images}
  Let \(f \colon A \to C\) and \(g \colon B \to C\) be two functions.
  The union of the images of \(f\) and \(g\) is the image of the coproduct map \(\coproductmap{f}{g} \colon A + B \to C\), i.e. \(\image(f) \union \image(g) = \image(\coproductmap{f}{g})\).
  The intersection of the images of \(f\) and \(g\) is the image of the pullback map \(\pullbackmap{f}{g} \colon A \times_{C} B \to C\), i.e. \(\image(f) \intersection \image(g) = \image(\pullbackmap{f}{g})\).
\end{rem}

\begin{rem}\label{rem:glueing-property}
  Let \(f \colon A \to C\), \(g \colon B \to C\) and \(\phi \colon C \to V\) such that \(\forall \ c \neq c' \in C \ \phi(c) = \phi(c') \implies c,c' \in \image(f)\).
  We have that \(\image(\pullbackmap{f \dcomp \phi}{g \dcomp \phi}) \supseteq \image(\pullbackmap{f}{g} \dcomp \phi)\).
  Then, \(\image(\pullbackmap{f \dcomp \phi}{g \dcomp \phi}) = \image(\pullbackmap{f}{g} \dcomp \phi)\) because their difference is empty:
  \begin{align*}
    & \image(\pullbackmap{f}{g} \dcomp \phi) \setminus \image(\pullbackmap{f \dcomp \phi}{g \dcomp \phi}) \\
    & = \{v \in V : \exists a \in A \ \exists b \in B \ \phi(f(a)) = \phi(g(b)) \land f(a) \notin \image(g) \land g(b) \notin \image(f)\} = \emptyset
  \end{align*}
\end{rem}

\begin{prop}\label{prop:mwd-bwd-lower}
  Let \(h = \cospan{A}{\partial_A}{H}{\partial_B}{B}\) with \(H = \mathgraph{F}{W}\).
  Let \(\phi \colon W \to V\) such that \(\forall \ w \neq w' \in W \ \phi(w) = \phi(w') \implies w,w' \in \image(\partial_A) \union \image(\partial_B)\) (glueing property).
  Let \(d\) be a monoidal decomposition of \(h\).
  Let \(\Gamma \defn (\mathgraph{F}{\image(\phi)},\image(\partial_A \dcomp \phi) \union \image(\partial_B \dcomp \phi))\).
  Then, there is an inductive branch decomposition \(\mTobdec(d)\) of \(\Gamma\) such that \(\decwidth(\mTobdec(d)) \leq 2 \cdot \max \{\decwidth(d), \card{A}, \card{B}\}\).
\end{prop}
\begin{proof}
  Proceed by induction on the decomposition tree \(d\).
  If it is just a leaf, \(d = \leafgenerator{h}\) and \(H\) has no edges, \(F = \emptyset\), then the corresponding inductive branch decomposition is empty, \(\mTobdec(d) \defn \emptydec\), and we can compute its width: \(\decwidth(\mTobdec(d)) \defn 0 \leq 2 \cdot \max \{\decwidth(d), \card{A}, \card{B}\}\).

  If the decomposition is just a leaf \(d = \leafgenerator{h}\) but \(H\) has exactly one edge, \(F = \{e\}\), then the corresponding branch decomposition is just a leaf as well, \(\mTobdec(d) \defn \leafgenerator{\Gamma}\), and we can compute its width: \(\decwidth(\mTobdec(d)) \defn \card{\image(\partial_A \dcomp \phi) \union \image(\partial_B \dcomp \phi)} \leq \card{A} + \card{B} \leq 2 \cdot \max \{\decwidth(d), \card{A}, \card{B}\}\).

  If the decomposition is just a leaf \(d = \leafgenerator{h}\) and \(H\) has more than one edge, \(\card{F} > 1\), then we can let \(\mTobdec(d)\) be any inductive branch decomposition of \(\Gamma\).
  Its width is not greater than the number of vertices in \(\Gamma\), thus we can bound its width \(\decwidth(\mTobdec(d)) \leq \card{\image(\phi)} \leq 2 \cdot \max\{\decwidth(d), \card{A}, \card{B}\}\).

  If \(d = \nodegenerator{d_1}{\dcomp_C}{d_2}\), then \(d_i\) is a monoidal decomposition of \(h_i\) with \(h = h_1 \dcomp_C h_2\).
  We can give the expressions of these morphisms: \(h_1 = \cospan{A}{\partial^1_A}{H_1}{\partial_1}{C}\) and \(h_2 = \cospan{C}{\partial_2}{H_2}{\partial^2_B}{B}\), with \(H_i = \mathgraph{F_i}{W_i}\), and obtain the following diagram, where \(\inclusion[i] \colon W_i \to W\) are the functions induced by the pushout and we define \(\phi_i \defn \inclusion[i] \dcomp \phi\).
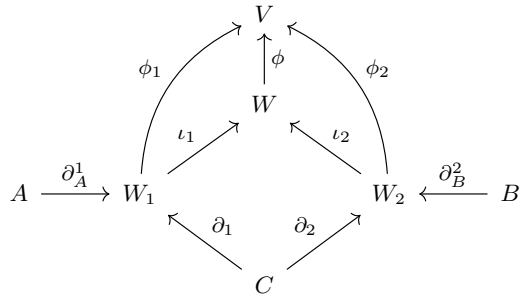
  \[\mondectobranchdecProofDiagram{}\]
  We show that \(\phi_1\) satisfies the glueing property in order to apply the induction hypothesis to \(\phi_{1}\) and \(H_{1}\):
  let \(w \neq w' \in W_1\) such that \(\phi_1(w) = \phi_1(w')\).
  Then, \(\inclusion[1](w) = \inclusion[1](w')\) or \(\phi(\inclusion[1](w)) = \phi(\inclusion[1](w')) \land \inclusion[1](w) \neq \inclusion[1](w')\).
  Then, \(w,w' \in \image(\partial_1)\) or \(\inclusion[1](w), \inclusion[1](w') \in \image(\partial_A \dcomp \phi) \union \image(\partial_B \dcomp \phi)\).
  Then, \(w,w' \in \image(\partial_1)\) or \(w,w' \in \image(\partial^1_A)\).
  Then, \(w,w' \in \image(\partial_1) \union \image(\partial^1_A)\).
  Similarly, we can show that \(\phi_2\) satisfies the same property.
  Then, we can apply the induction hypothesis to get an inductive branch decomposition \(\mTobdec(d_1)\) of \(\Gamma_1 = (\mathgraph{F_1}{\image(\phi_1)}, \image(\partial^1_A \dcomp \phi_1) \union \image(\partial_1 \dcomp \phi_1))\) and an inductive branch decomposition \(\mTobdec(d_2)\) of \(\Gamma_2 = (\mathgraph{F_2}{\image(\phi_2)}, \image(\partial^2_B \dcomp \phi_2) \union \image(\partial_2 \dcomp \phi_2))\) with bounded width: \(\decwidth(\mTobdec(d_1)) \leq 2 \cdot \max\{\decwidth(d_1),\card{A},\card{C}\}\) and  \(\decwidth(\mTobdec(d_2)) \leq 2 \cdot \max\{\decwidth(d_2),\card{B},\card{C}\}\).

  We check that we can define an inductive branch decomposition of \(\Gamma\) from \(\mTobdec(d_1)\) and \(\mTobdec(d_2)\).
  \begin{itemize}
    \item \(F = F_1 \disjointunion F_2\) because the pushout is along discrete hypergraphs.
    \item \(\image(\phi) = \image(\phi_1) \union \image(\phi_2)\) because \(\image(\coproductmap{\inclusion[1]}{\inclusion[2]}) = W\) and \(\image(\phi_1) \union \image(\phi_2) = \image(\inclusion[1] \dcomp \phi) \union \image(\inclusion[2] \dcomp \phi) = \image(\coproductmap{\inclusion[1]}{\inclusion[2]} \dcomp \phi) = \image(\phi)\).
    \item \(\image(\coproductmap{\partial^1_A}{\partial_1} \dcomp \phi_1) = \image(\phi_1) \intersection (\image(\phi_2) \union \image(\partial_A \dcomp \phi) \union \image(\partial_B \dcomp \phi))\) because
      \begin{align*}
        & \image(\phi_1) \intersection (\image(\phi_2) \union \image(\partial_A \dcomp \phi) \union \image(\partial_B \dcomp \phi)) \\
        \explain{by definition of \(\phi_i\)}\\
        & \image(\inclusion[1] \dcomp \phi) \intersection (\image(\inclusion[2] \dcomp \phi) \union \image(\partial_A \dcomp \phi) \union \image(\partial_B \dcomp \phi)) \\
        \explain{because \(\image(\partial_B) = \image(\partial^2_B \dcomp \inclusion[2]) \subseteq \image(\inclusion[2])\)}\\
        & \image(\inclusion[1] \dcomp \phi) \intersection (\image(\inclusion[2] \dcomp \phi) \union \image(\partial_A \dcomp \phi)) \\
        \explain{by \Cref{rem:compute-images}}\\
        & \image(\inclusion[1] \dcomp \phi) \intersection \image(\coproductmap{\inclusion[2]}{\partial_A} \dcomp \phi) \\
        \explain{by \Cref{rem:compute-images}}\\
        & \image(\pullbackmap{\inclusion[1] \dcomp \phi}{\coproductmap{\inclusion[2]}{\partial_A} \dcomp \phi})\\
        \explain{by \Cref{rem:glueing-property}}\\
        & \image(\pullbackmap{\inclusion[1]}{\coproductmap{\inclusion[2]}{\partial_A}} \dcomp \phi)\\
        \explain{because pullbacks commute with coproducts}\\
        & \image(\coproductmap{\pullbackmap{\inclusion[1]}{\inclusion[2]}}{\pullbackmap{\inclusion[1]}{\partial_A}} \dcomp \phi)\\
        \explain{because \(\partial_A = \partial^1_A \dcomp \inclusion[1]\)}\\
        & \image(\coproductmap{\pullbackmap{\inclusion[1]}{\inclusion[2]}}{\partial_A} \dcomp \phi)\\
        \explain{because \(\partial_1 \dcomp \inclusion[1] = \partial_2 \dcomp \inclusion[2]\) is the pushout map of \(\partial_1\) and \(\partial_2\)}\\
        & \image(\coproductmap{\partial_1 \dcomp \inclusion[1]}{\partial^1_A \dcomp \inclusion[1]} \dcomp \phi)\\
        \explain{by property of the coproduct}\\
        & \image(\coproductmap{\partial_1}{\partial^1_A} \dcomp \phi_1)
      \end{align*}
    \item \(\image(\coproductmap{\partial_2}{\partial^2_B} \dcomp \phi_2) = \image(\phi_2) \intersection (\image(\phi_1) \union \image(\partial_A \dcomp \phi) \union \image(\partial_B \dcomp \phi))\) similarly to the former point.
  \end{itemize}
  Then, \(\mTobdec(d) \defn \nodegenerator{\mTobdec(d_1)}{\Gamma}{\mTobdec(d_2)}\) is an inductive branch decomposition of \(\Gamma\) and
  \begin{align*}
    & \decwidth(\mTobdec(d)) \\
    & \defn \max\{\decwidth(\mTobdec(d_1)), \card{\image(\coproductmap{\partial_A}{\partial_B})}, \decwidth(\mTobdec(d_2)) \}\\
    & \leq \max\{2 \cdot \decwidth(d_1), 2 \cdot \card{A}, 2 \cdot \card{C}, \card{A} + \card{B},\\
    & \qquad 2 \cdot \decwidth(d_2), 2 \cdot \card{B} \}\\
    & \leq 2 \cdot \max\{\decwidth(d_1), \card{A}, \card{C}, \decwidth(d_2), \card{B} \}\\
    & \codefn 2 \cdot \max\{\decwidth(d), \card{A}, \card{B} \}
  \end{align*}

  If \(d = \nodegenerator{d_1}{\tensor}{d_2}\), then \(d_i\) is a monoidal decomposition of \(h_i\) with \(h = h_1 \tensor h_2\).
  Let \(h_i = \cospan{X_i}{\partial^i_X}{H_i}{\partial^i_Y}{Y_i}\) with \(H_i = \ctgraph{F_i}{W_i}\).
  Let \(\inclusion[i] \colon W_i \to W\) be the inclusions induced by the monoidal product.
  Define \(\phi_i \defn \inclusion[i] \dcomp \phi\).
  We show that \(\phi_1\) satisfies the glueing property:
  Let \(w \neq w' \in W_1\) such that \(\phi_1(w) = \phi_1(w')\).
  Then, \(\inclusion[1](w) = \inclusion[1](w')\) or \(\phi(\inclusion[1](w)) = \phi(\inclusion[1](w')) \land \inclusion[1](w) \neq \inclusion[1](w')\).
  Then, \(\inclusion[1](w), \inclusion[1](w') \in  \image(\partial_A \dcomp \phi) \union \image(\partial_B \dcomp \phi)\) because \(\inclusion[i]\) are injective.
  Then, \(w,w' \in \image(\partial^1_A) \union \image(\partial^1_B)\).
  Similarly, we can show that \(\phi_2\) satisfies the same property.
  Then, we can apply the induction hypothesis to get \(\mTobdec(d_i)\) inductive branch decomposition of \(\Gamma_i = (\mathgraph{F_i}{\image(\phi_i)}, \image(\coproductmap{\partial^i_A}{\partial^i_B} \dcomp \phi_i))\) such that \(\decwidth(\mTobdec(d_i)) \leq 2 \cdot \max\{\decwidth(d_i),\card{A_i},\card{B_i}\}\).

  We check that we can define an inductive branch decomposition of \(\Gamma\) from \(\mTobdec(d_1)\) and \(\mTobdec(d_2)\).
  \begin{itemize}
    \item \(F = F_1 \disjointunion F_2\) because the monoidal product is given by the coproduct in \(\Set\).
    \item \(\image(\phi) = \image(\phi_1) \union \image(\phi_2)\) because \(\image(\coproductmap{\inclusion[1]}{\inclusion[2]}) = W\) and \(\image(\phi_1) \union \image(\phi_2) = \image(\inclusion[1] \dcomp \phi) \union \image(\inclusion[2] \dcomp \phi) = \image(\coproductmap{\inclusion[1]}{\inclusion[2]} \dcomp \phi) = \image(\phi)\).
    \item \(\image(\coproductmap{\partial^1_A}{\partial^1_B} \dcomp \phi_1) = \image(\phi_1) \intersection (\image(\phi_2) \union \image(\partial_A \dcomp \phi) \union \image(\partial_B \dcomp \phi))\) because
      \begin{align*}
        & \image(\phi_1) \intersection (\image(\phi_2) \union \image(\partial_A \dcomp \phi) \union \image(\partial_B \dcomp \phi)) \\
        \explain{by definition of \(\phi_i\)}\\
        & \image(\inclusion[1] \dcomp \phi) \intersection (\image(\inclusion[2] \dcomp \phi) \union \image(\partial_A \dcomp \phi) \union \image(\partial_B \dcomp \phi)) \\
        \explain{by \Cref{rem:compute-images} and property of the coproduct}\\
        & \image(\inclusion[1] \dcomp \phi) \intersection \image(\coproductmap{\inclusion[2]}{\coproductmap{\partial_A}{\partial_B}} \dcomp \phi)\\
        \explain{by \Cref{rem:compute-images}}\\
        & \image(\pullbackmap{\inclusion[1] \dcomp \phi}{\coproductmap{\inclusion[2]}{\coproductmap{\partial_A}{\partial_B}} \dcomp \phi})\\
        \explain{by \Cref{rem:glueing-property}}\\
        & \image(\pullbackmap{\inclusion[1]}{\coproductmap{\inclusion[2]}{\coproductmap{\partial_A}{\partial_B}}} \dcomp \phi)\\
        \explain{because pullbacks commute with coproducts}\\
        & \image(\coproductmap{\pullbackmap{\inclusion[1]}{\inclusion[2]}}{\pullbackmap{\inclusion[1]}{\coproductmap{\partial_A}{\partial_B}}} \dcomp \phi)\\
        \explain{because \(\pullbackmap{\inclusion[1]}{\inclusion[2]} = \initmap{}\)}\\
        & \image(\pullbackmap{\inclusion[1]}{\coproductmap{\partial_A}{\partial_B}} \dcomp \phi)\\
        \explain{because \(\partial_A = \partial^1_A + \partial^2_A\) and \(\partial_B = \partial^1_B + \partial^2_B\)}\\
        & \image(\coproductmap{\partial^1_A \dcomp \inclusion[1]}{\partial^1_B \dcomp \inclusion[1]} \dcomp \phi)\\
        \explain{by property of the coproduct}\\
        & \image(\coproductmap{\partial^1_A}{\partial^1_B} \dcomp \phi_1)
      \end{align*}
    \item \(\image(\coproductmap{\partial^2_A}{\partial^2_B} \dcomp \phi_2) = \image(\phi_2) \intersection (\image(\phi_1) \union \image(\partial_A \dcomp \phi) \union \image(\partial_B \dcomp \phi))\) similarly to the former point.
  \end{itemize}
  Then, \(\mTobdec(d) \defn \nodegenerator{\mTobdec(d_1)}{\Gamma}{\mTobdec(d_2)}\) is an inductive branch decomposition of \(\Gamma\) and
  \begin{align*}
    & \decwidth(\mTobdec(d)) \\
    & \defn \max\{\decwidth(\mTobdec(d_1)), \card{\image(\coproductmap{\partial_A}{\partial_B})}, \decwidth(\mTobdec(d_2)) \}\\
    & \leq \max\{2 \cdot \decwidth(d_1), 2 \cdot \card{A_1}, 2 \cdot \card{B_1}, \card{A} + \card{B},\\
    & \qquad 2 \cdot \decwidth(d_2), 2 \cdot \card{A_2}, 2 \cdot \card{B_2} \}\\
    & \leq 2 \cdot \max\{\decwidth(d_1), \card{A}, \decwidth(d_2), \card{B} \}\\
    & \codefn 2 \cdot \max\{\decwidth(d), \card{A}, \card{B} \}
  \end{align*}
  where we applied the induction hypothesis and \Cref{def:weight-cospan-graph}.
\end{proof}

Combining \Cref{th:twd-bwd}, \Cref{prop:rec-branch-width-equivalent}, \Cref{prop:mwd-bwd-upper}, and \Cref{prop:mwd-bwd-lower}, we obtain the following.

\begin{thm}\label{th:mwd-bwd}
  Branch width is equivalent to monoidal width in \(\cospanUHGraphO\).
  More precisely, let \(G\) be a hypergraph and \(g = \cospan{\emptyset}{}{G}{}{\emptyset}\) be the corresponding morphism of \(\cospanUHGraphO\).
  Then, \(\frac{1}{2} \cdot \branchwidth(G) \leq \mwd(g) \leq \branchwidth(G) + 1\).
\end{thm}

With \Cref{th:twd-bwd}, we obtain:

\begin{cor}
  Tree width is equivalent to monoidal width in \(\cospanUHGraphO\).
\end{cor}

\section{Monoidal width in matrices}\label{sec:mwd-matrices}
We have just seen that instantiating monoidal width in a monoidal category of graphs yields a measure that is equivalent to tree width.
Now, we turn our attention to rank width, which is more linear algebraic in flavour as it relies on treating the connectivity of graphs by means of adjacency matrices.
Thus, the monoidal category of matrices is a natural example to study first.
We relate monoidal width in the category of matrices over the natural numbers, which we introduce in \Cref{sec:prop-matrices}, to their rank (\Cref{sec:mwd-matrices-proof}).

The rank of a matrix is the maximum number of its linearly independent rows (or, equivalently, columns).
Conveniently, it can be characterised in terms of minimal factorisations.

\begin{lemC}[\cite{piziak1999fullRank}]\label{lemma:min-cut-is-rank}
  Let \(A \in \MatN(m,n)\) be an \(m\) by \(n\) matrix with entries in the natural numbers.
  Then \(\rank(A) = \min\{k \in \naturals \st \exists B \in \MatN(k,n) \ \exists C \in \MatN(m,k) \ A = C \cdot B\}\).
\end{lemC}

\subsection{The prop of matrices}\label{sec:prop-matrices}
The monoidal category \(\Mat_{\naturals}\) of matrices with entries in the natural numbers is a prop whose morphisms from \(n\) to \(m\) are \(m\) by \(n\) matrices.

\begin{defi}
  \(\MatN\) is the prop whose morphisms \(n \to m\) are \(m\) by \(n\) matrices with entries in the natural numbers.
  Composition is the usual product of matrices and the monoidal product is the biproduct \(A \tensor B \defn \begin{psmallmatrix} A & \zeromat \\ \zeromat & B \end{psmallmatrix}\).
\end{defi}

Let us examine matrix decompositions enabled by this algebra.
A matrix \(A\) can be written as a monoidal product \(A = A_{1} \tensor A_{2}\) iff the matrix has blocks \(A_{1}\) and \(A_{2}\), i.e. \(A = \begin{psmallmatrix} A_{1} & \zeromat \\ \zeromat & A_{2}\end{psmallmatrix}\).
On the other hand, a composition is related to the rank: the statement of \Cref{lemma:min-cut-is-rank} can be read in the category \(\MatN\) as \(\rank(A) = \min\{k \in \naturals : A = B \dcomp_k C\}\).

\begin{thmC}[\cite{zanasi2015thesis}]
  Let \(\Bialg\) be the prop whose generators and axioms are given in \Cref{fig:bialgebra}.
  There is an isomorphism of categories \(\BialgToMat \colon \Bialg \to \MatN\).
\end{thmC}

\begin{figure}[h!]
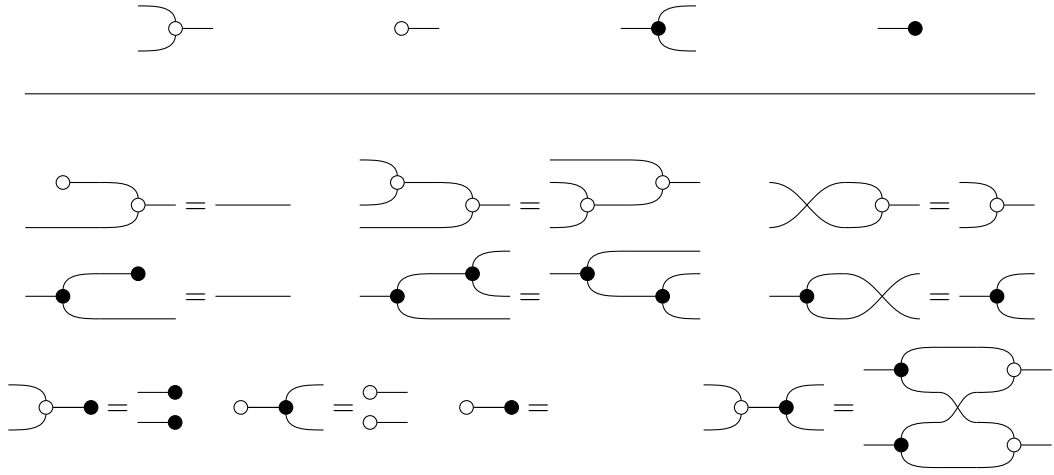

  \bialgebraFig{}
  \caption{Generators and axioms of a bialgebra.}\label{fig:bialgebra}
\end{figure}

Every morphism \(f \colon n \to m\) in \(\Bialg\) corresponds to a matrix \(A = \BialgToMat(f) \in \Mat_{\naturals}(m,n)\): we can read the \((i,j)\)-entry of \(A\) off the diagram of \(f\) by counting the number of  paths from the \(j\)th input to the \(i\)th output.

\begin{exa}
  The matrix \(\begin{psmallmatrix} 1 & 0 \\ 1 & 2 \\ 0 & 0 \end{psmallmatrix} \in \MatN(3,2)\) corresponds to
  \begin{center}
    \matrixExFig{}
  \end{center}
  For matrices \(A \in \MatN(m,n)\), \(B \in \MatN(m,p)\) and \(C \in \MatN(l,n)\), we indicate with \((A \mid B) \in \MatN(m, n+p)\) and with \(\begin{psmallmatrix}A \\ C\end{psmallmatrix} \in \MatN(m+l,n)\) the matrices obtained by concatenating \(A\) with \(B\) horizontally or with \(C\) vertically.
\end{exa}

In order to instantiate monoidal width in \(\Bialg\), we need to define an appropriate weight function: the natural choice for a prop is to assign weight \(n\) to compositions along the object \(n\).

\begin{defi}\label{def:weight-bialgebra}
  The atoms for \(\Bialg\) are its generators (\Cref{fig:bialgebra}) with the symmetry and identity on \(1\): \(\decgenerators = \{\cp_1, \delete_1, \add_1, \zero_1, \swap{1,1},\id{1}\}\).
  The weight function \(\nodeweight \colon \decgenerators \union \monoidaloperations{\Bialg} \to \naturals\) has \(\nodeweight(n) \defn n\), for any \(n \in \naturals\), and \(\nodeweight(g) \defn \max\{m,n\}\), for \(g \colon n \to m \in \decgenerators\).
\end{defi}

\subsection{Monoidal width of matrices}\label{sec:mwd-matrices-proof}
We show that the monoidal width of a morphism in the category of matrices \(\Bialg\), with the weight function in \Cref{def:weight-bialgebra}, is, up to \(1\), the maximum rank of its blocks.
The overall strategy to prove this result is to first relate monoidal width directly with the rank (\Cref{prop:matrix-mwd-non-tensor-decomposables}) and then to improve this bound by prioritising \(\tensor\)-nodes in a decomposition (\Cref{prop:matrices-better-tensors-first}).
Combining these two results leads to \Cref{th:mwd-matrices}.
The shape of an optimal decomposition is given in \Cref{fig:matrix-decomposition}: a matrix \(A = \begin{psmallmatrix}A_{1} & \zeromat & \cdots & \zeromat\\ \zeromat & A_{2} & \cdots & \zeromat \\ \vdots & \vdots & \ddots \vdots \\ \zeromat & \zeromat & \cdots & A_{k}\end{psmallmatrix}\) can be decomposed as \(A = (M_{1} \dcomp N_{1}) \tensor (M_{2} \dcomp N_{2}) \tensor \cdots \tensor (M_{k} \dcomp N_{k})\), where \(A_{j} = M_{j} \dcomp N_{j}\) is a rank factorisation as in \Cref{lemma:min-cut-is-rank}.

\begin{figure}[h!]
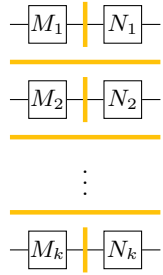

  \matrixDecompositionShapeFig{}
  \caption{Generic shape of an optimal decomposition in \(\Bialg\).}\label{fig:matrix-decomposition}
\end{figure}

The characterisation of the rank of a matrix in \Cref{lemma:min-cut-is-rank} hints at some relationship between the monoidal width of a matrix and its rank.
In fact, we have \Cref{prop:matrix-mwd-non-tensor-decomposables}, which bounds the monoidal width of a matrix with its rank.
In order to prove this result, we first need to bound the monoidal width of a matrix with its domain and codomain, which is done in \Cref{prop:mwd-lessthan-domain-codomain}.

\begin{prop}\label{prop:mwd-lessthan-domain-codomain}
  Let \(\cat{P}\) be a cartesian and cocartesian prop.
  Suppose that $\id{1}, \cp_1, \add_1, \delete_1,$ $\zero_1 \in \decgenerators$
  and \(\nodeweight(\id{1}) \leq 1\), \(\nodeweight(\cp_1) \leq 2\), \(\nodeweight(\add_1) \leq 2\), \(\nodeweight(\delete_1) \leq 1\) and \(\nodeweight(\zero_1) \leq 1\).
  Suppose that, for every \(g \colon 1 \to 1\), \(\mwd(g) \leq 2\).
  Let \(f \colon n \to m\) be a morphism in \(\cat{P}\).
  Then \(\mwd(f) \leq \min\{m,n\}+1\).
\end{prop}
\begin{proof}
  We proceed by induction on \(k = \max\{m,n\}\).
  There are three base cases.
  \begin{itemize}
    \item If \(n = 0\), then \(f = \zero_m\) because \(0\) is initial by hypothesis, and we can compute its width, \(\mwd(f) = \mwd(\Tensor_m \zero_1) \leq \nodeweight(\zero_1) \leq 1 \leq 0 + 1 \).
    \item If \(m = 0\), then \(f = \delete_n\) because \(0\) is terminal by hypothesis, and we can compute its width, \(\mwd(f) = \mwd(\Tensor_m \delete_1) \leq \nodeweight(\delete_1) \leq 1 \leq 0 + 1 \).
    \item If \(m = n = 1\), then \(\mwd(f) \leq 2 \leq 1+1\) by hypothesis.
  \end{itemize}
  For the induction steps, suppose that the statement is true for any \(f' \colon n' \to m'\) with \(\max\{m',n'\} < k = \max\{m,n\}\) and \(\min \{m',n'\} \geq 1\).
  There are three possibilities.
  \begin{enumerate}
    \item\label{item:1-proof} If \(0 < n < m = k\), then \(f\) can be decomposed as shown below because \(\cp_{n+1}\) is uniform and morphisms are copiable because \(\cat{P}\) is cartesian by hypothesis.
      \[\mwdMatlessthanBoundariesOne{}\]
      This corresponds to \(f = \cp_n \dcomp (\id{n} \tensor h_1) \dcomp_{n+1} (h_2 \tensor \id{1})\), where \(h_1 \defn f \dcomp (\delete_{m-1} \tensor \id{1}) \colon n \to 1\) and \(h_2 \defn f \dcomp (\id{m-1} \tensor \delete_1) \colon n \to m-1\).

      Then, \(\mwd(f) \leq \max\{\mwd(\cp_n \dcomp (\id{n} \tensor h_1)),n+1,\mwd(h_2 \tensor \id{1})\}\).
      So, we want to bound the monoidal width of the two morphisms appearing in the formula above.
      For the first morphism, we apply the induction hypothesis because \(h_1 \colon n \to 1\) and \(1 , n < k\).
      For the second morphism, we apply the induction hypothesis because \(h_2 \colon n \to m-1\) and \(n, m-1 < k\).
      \begin{align*}
        & \mwd(\cp_n \dcomp (\id{n} \tensor h_1)) && \mwd(h_2 \tensor \id{1}) \\
        \explain[\leq]{by \Cref{lemma:mwd-copy}} & \explain{by \Cref{def:mwd}}\\
        & \max\{\mwd(h_1),n+1\} && \mwd(h_2) \\
        \explain[\leq]{by induction hypothesis} & \explain[\leq]{by induction hypothesis}\\
        & \max\{\min\{n,1\}+1,n+1\} && \min\{n,m-1\}+1\\
        \explain{because \(0 < n\)} & \explain{because \(n \leq m-1\)}\\
        & n+1 && n+1
      \end{align*}
      Then, \(\mwd(f) \leq n+1 = \min\{m,n\}+1\) because \(n<m\).
    \item\label{item:2-proof} If \(0<m < n=k\), we can apply \Cref{item:1-proof} to \(\cat{P}\op\) with the same assumptions on the set of atoms because \(\cat{P}\op\) is also cartesian and cocartesian.
      We obtain that \(\mwd(f) \leq m + 1 = \min\{m,n\}+1\) because \(m < n\).
    \item If \(0 < m = n = k\), \(f\) can be decomposed as in \Cref{item:1-proof} and, instead of applying the induction hypothesis to bound \(\mwd(h_1)\) and \(\mwd(h_2)\), one applies \Cref{item:2-proof}.
      Then, \(\mwd(f) \leq m+1 = \min\{m,n\}+1\) because \(m = n\).
      \qedhere
  \end{enumerate}
\end{proof}

We can apply the former result to \(\Bialg\) and obtain \Cref{prop:matrix-mwd-non-tensor-decomposables} because the width of \(1 \times 1\) matrices, which are numbers, is at most \(2\).
This follows from the reasoning in~\Cref{ex:mwd-number-like-morphisms} as we can write every natural number \(k \colon 1 \to 1\) as the following composition:
\begin{center}
  \mondecnumbersFig{}
\end{center}
\begin{lem}\label{lemma:mwd-numbers}
  Let \(k \colon 1 \to 1\) in \(\Bialg\).
  Then, \(\mwd(k) \leq 2\).
\end{lem}

\begin{prop}\label{prop:matrix-mwd-non-tensor-decomposables}
  Let \(f \colon n \to m\) in \(\Bialg\).
  Then, \(\mwd f \leq \rank(\BialgToMat f) +1\).
  Moreover, if \(f\) is not \(\tensor\)-decomposable, i.e.~there are no \(f_1,f_2\) both distinct from \(f\) s.t. \(f = f_1 \tensor f_2\), then \(\rank(\BialgToMat f) \leq \mwd f\).
\end{prop}
\begin{proof}
  We prove the second inequality.
  Let \(d\) be a monoidal decomposition of \(f\).
  By hypothesis, \(f\) is non \(\tensor\)-decomposable.
  Then, there are two options.
  \begin{enumerate}
    \item If the decomposition is just a leaf, \(d = \leafgenerator{f}\), then \(f\) must be an atom.
          We can check the inequality for all the atoms: \(\nodeweight(\swap{}) = 2 \geq \rank(\BialgToMat f) = 2\), \(\nodeweight(\cp_1) = \nodeweight(\add_1) = 2  \geq \rank(\BialgToMat f) = 1\) or \(\nodeweight(\delete_1) = \nodeweight(\zero_1) = 1 \geq \rank(\BialgToMat f) = 0\).
          Then, \(\decwidth(d) = \nodeweight(f) \geq \rank(\BialgToMat f)\).
    \item If \(d = \nodegenerator{d_1}{\dcomp_k}{d_2}\), then there are \(g \colon n \to k\) and \(h \colon k \to m\) such that \(f = g \dcomp h\).
      By \Cref{lemma:min-cut-is-rank}, \(k \geq \rank(\BialgToMat f)\).
      Then, \(\decwidth(d) \geq k \geq \rank(\BialgToMat f)\).
  \end{enumerate}
  We prove the first inequality.
  By \Cref{lemma:min-cut-is-rank}, there are \(g \colon n \to r\) and \(h \colon r \to m\) such that \(f = g \dcomp h\) with \(r = \rank(\BialgToMat f)\).
  Then, \(r \leq m,n\) by definition of rank.
  By \Cref{lemma:mwd-numbers}, we can apply \Cref{prop:mwd-lessthan-domain-codomain} to obtain that \(\mwd(g) \leq \min\{n,r\}+1 = r+1\) and \(\mwd(h) \leq \min\{m,r\}+1 = r+1\).
  Then, \(\mwd(f) \leq \max\{\mwd(g),r,\mwd(h)\} \leq r+1\).
\end{proof}

The bounds given by \Cref{prop:matrix-mwd-non-tensor-decomposables} can be improved when we have a \(\tensor\)-decomposition of a matrix, i.e. we can write \(f = f_{1} \tensor \dots \tensor f_{k}\), to obtain \Cref{prop:matrices-better-tensors-first}.
The latter relies on \Cref{lemma:mwd-after-discard}, which shows that discarding inputs or outputs cannot increase the monoidal width of a morphism in \(\Bialg\).

\begin{lem}\label{lemma:mwd-after-discard}
  Let \(f \colon n \to m\) in \(\Bialg\) and \(d \in \decset{f}\).
  Let \(f_{D} \defn f \dcomp (\id{m-k} \tensor \delete_k)\) and \(f_{Z} \defn (\id{n-k'} \tensor \zero_{k'}) \dcomp f\), with \(k \leq m\) and \(k' \leq n\).
  \[\mwdafterdiscardStateFig{}\]
  Then there are \(\deleteMdec(d) \in \decset{f_{D}}\) and \(\codeleteMdec(d) \in \decset{f_{Z}}\) such that \(\decwidth(\deleteMdec(d)) \leq \decwidth(d)\) and \(\decwidth(\codeleteMdec(d)) \leq \decwidth(d)\).
\end{lem}
\begin{proof}
  We show the inequality for \(f_{D}\) by induction on the decomposition \(d\).
  The inequality for \(f_{Z}\) follows from the fact that \(\Bialg\) coincides with its opposite category.
  If the decomposition has only one node, \(d = \leafgenerator{f}\), then \(f\) is an atom and we can check these cases by hand in the table below.
  The first column shows the possibilities for \(f\), while the second and third columns show the decompositions of \(f_{D}\) for \(k = 1\) and \(k=2\).
  \[\mwdafterdiscardProofFigOne{}\]

  If the decomposition starts with a composition node, \(d = \nodegenerator{d_{1}}{\dcomp}{d_{2}}\), then \(f = f_1 \dcomp f_2\), with \(d_{i}\) monoidal decomposition of \(f_i\).
  \[\mwdafterdiscardProofFigTwo{}\]
  By induction hypothesis, there is a monoidal decomposition \(\deleteMdec(d_2)\) of \(f_2 \dcomp (\id{m-k} \tensor \delete_k)\) such that \(\decwidth(\deleteMdec(d_2)) \leq \decwidth(d_2)\).
  We use this decomposition to define a decomposition \(\deleteMdec(d) \defn \nodegenerator{d_1}{\dcomp}{\deleteMdec(d_2)}\) of \(f_{D}\).
  Then, \(\deleteMdec(d)\) is a monoidal decomposition of \(f \dcomp (\id{m-k} \tensor \delete_k)\) because \(f \dcomp (\id{m-k} \tensor \delete_k) = f_1 \dcomp f_2 \dcomp (\id{m-k} \tensor \delete_k)\).

  If the decomposition starts with a tensor node, \(d = \nodegenerator{d_1}{\tensor}{d_2}\), then \(f = f_1 \tensor f_2\), with \(d_i\) monoidal decomposition of \(f_i \colon n_i \to m_i\).
  There are two possibilities: either \(k \leq m_{2}\) or \(k > m_{2}\).
  If \(k \leq m_2\), then \(f \dcomp (\id{m-k} \tensor \delete_k) = f_1 \tensor (f_2 \dcomp (\id{m_2-k} \tensor \delete_k))\).
  \[\mwdafterdiscardProofFigThree{}\]
  By induction hypothesis, there is a monoidal decomposition \(\deleteMdec(d_2)\) of \(f_2 \dcomp (\id{m-k} \tensor \delete_k)\) such that \(\decwidth(\deleteMdec(d_2)) \leq \decwidth(d_2)\).
  Then, we can use this decomposition to define a decomposition \(\deleteMdec(d) \defn \nodegenerator{d_1}{\tensor}{\deleteMdec(d_2)}\) of \(f_{D}\).
  If \(k > m_2\), then \(f \dcomp (\id{m-k} \tensor \delete_k) = (f_1 \dcomp (\id{m_1 - k + m_2} \tensor \delete_{k-m_2})) \tensor (f_2 \dcomp \delete_{m_2})\).
  \[\mwdafterdiscardProofFigFour{}\]
  By induction hypothesis, there are monoidal decompositions \(\deleteMdec(d_i)\) of \(f_1 \dcomp (\id{m_1-k+m_2} \tensor \delete_{k-m_2})\) and \(f_2 \dcomp \delete_{m_2}\) such that \(\decwidth(\deleteMdec(d_i)) \leq \decwidth(d_i)\).
  Then, we can use these decompositions to define a monoidal decomposition \(\deleteMdec(d) \defn \nodegenerator{\deleteMdec(d_1)}{\tensor}{\deleteMdec(d_2)}\) of \(f_{D}\).
\end{proof}

\begin{prop}\label{prop:matrices-better-tensors-first}
  Let \(f \colon n \to m\) in \(\Bialg\) and \(d' = \nodegenerator{d'_1}{\dcomp_k}{d'_2} \in \decset{f}\).
  Suppose there are \(f_1\) and \(f_2\) such that \(f = f_1 \tensor f_2\).
  Then, there is \(d = \nodegenerator{d_1}{\tensor}{d_2} \in \decset{f}\) such that \(\decwidth(d) \leq \decwidth(d')\).
\end{prop}
\begin{proof}
  By hypothesis, \(d'\) is a monoidal decomposition of \(f\).
  Then, there are \(g\) and \(h\) such that \(f_1 \tensor f_2 = f = g \dcomp h\).
  By \Cref{prop:matrix-mwd-non-tensor-decomposables}, there are monoidal decompositions \(d_i\) of \(f_i\) with \(\decwidth(d_i) \leq r_i + 1\), where \(r_i \defn \rank(\BialgToMat f_i)\).
  By properties of the rank, \(r_1 + r_2 = \rank(\BialgToMat f)\) and, by \Cref{lemma:min-cut-is-rank}, \(\rank(\BialgToMat f) \leq k\).
  
  There are two cases: either both ranks are non-zero, or at least one is zero.
  If \(r_i > 0\), then \(r_1 + r_2 \geq \max\{r_1,r_2\}+1\).
  If there is \(r_i = 0\), then \(f_i = \delete \dcomp_0 \zero\) and we may assume that \(f_1 = \delete \dcomp_0 \zero\).
  Then, we can express \(f_2\) in terms of \(g\) and \(h\).
  \[\propmatricesbettertensorsfirstproofFig{}\]
  By \Cref{lemma:mwd-after-discard}, \(\mwd((\id{} \tensor \zero) \dcomp g) \leq \mwd(g)\) and \(\mwd(h \dcomp (\id{} \tensor \delete)) \leq \mwd(h)\).
  We compute the widths of the decompositions in these two cases.
  \begin{align*}
    &\text{Case }r_{i} > 0 && \text{Case }r_{1}=0\\
    & \decwidth(d') && \decwidth(d') \\
    & = \max\{\decwidth(d'_1), k, \decwidth(d'_2)\} && = \max\{\decwidth(d'_1), k, \decwidth(d'_2)\}\\
    & \geq k && \geq \max\{\mwd(g), k, \mwd(h)\} \\
    & \geq \rank(\BialgToMat f) && \geq \max\{\mwd((\id{} \tensor \zero) \dcomp g), k, \mwd(h \dcomp (\id{} \tensor \delete))\} \\
    & = r_1 + r_2 && \geq \mwd(f_2)\\
    & \geq \max\{r_1,r_2\}+1 && = \decwidth(d_2)\\
    & \geq \max\{\decwidth(d_1), \decwidth(d_2)\} && = \decwidth(d)\\
    & = \decwidth(d) &&
  \qedhere
  \end{align*}
\end{proof}

We summarise \Cref{prop:matrices-better-tensors-first} and \Cref{prop:matrix-mwd-non-tensor-decomposables} in \Cref{cor:mwd-matrices-blocks}.

\begin{cor}\label{cor:mwd-matrices-blocks}
  Let \(f = f_{1} \tensor \dots \tensor f_{k}\) in \(\Bialg\).
  Then, \(\mwd(f) \leq \max_{i = 1,\dots,k} \rank(\BialgToMat(f_{i})) + 1\).
  Moreover, if \(f_i\) are not \(\tensor\)-decomposable, then \(\max_{i = 1,\dots,k} \rank(\BialgToMat(f_{i})) \leq \mwd f\).
\end{cor}
\begin{proof}
  By \Cref{prop:matrices-better-tensors-first} there is a decomposition of \(f\) of the form \(d = \nodegenerator{d_1}{\tensor}{\cdots \nodegenerator{d_{k-1}}{\tensor}{d_k}}\), where we can choose \(d_i\) to be a minimal decomposition of \(f_i\).
  Then, \(\mwd(f) \leq \decwidth(d) = \max_{i = 1,\ldots,k} \decwidth(d_i)\).
  By \Cref{prop:matrix-mwd-non-tensor-decomposables}, \(\decwidth(d_i) \leq r_i +1\).
  Then, \(\mwd(f) \leq \max\{r_1,\ldots,r_k\}+1\).
  Moreover, if \(f_i\) are not \(\tensor\)-decomposable, \Cref{prop:matrix-mwd-non-tensor-decomposables} gives also a lower bound on their monoidal width: \(\rank(\BialgToMat(f_{i})) \leq \mwd f_i\); and we obtain that \(\max_{i = 1,\dots,k} \rank(\BialgToMat(f_{i})) \leq \mwd f\).
\end{proof}

The results so far show a way to construct efficient decompositions given a \(\tensor\)-decomposi\-tion of the matrix.
However, we do not know whether \(\tensor\)-decompositions are unique.
\Cref{prop:mon-cat-unique-tensor-decomposition} shows that every morphism in \(\Bialg\) has a unique \(\tensor\)-decomposition.

\begin{prop}\label{prop:mon-cat-unique-tensor-decomposition}
  Let \(\cat{C}\) be a monoidal category whose monoidal unit \(0\) is both initial and terminal, and whose objects are a unique factorisation monoid.
  Let \(f\) be a morphism in \(\cat{C}\).
  Then \(f\) has a unique \(\tensor\)-decomposition.
\end{prop}
\begin{proof}
  Suppose \(f = f_1 \tensor \cdots \tensor f_m = g_1 \tensor \cdots \tensor g_n\) with \(f_i \colon X_i \to Y_i\) and \(g_j \colon Z_j \to W_j\) non \(\tensor\)-decomposables.
  Suppose \(m \leq n\) and proceed by induction on \(m\).
  If \(m=0\), then \(f = \id{0}\) and \(g_{i} = \id{0}\) for every \(i = 1,\dots,n\) because \(0\) is initial and terminal.

  Suppose that \(\bar{f} \defn f_{1} \tensor \dots \tensor f_{m-1}\) has a unique \(\tensor\)-decomposition.
  Let \(A_{1} \tensor \dots \tensor A_{\alpha}\) and \(B_{1} \tensor \dots \tensor B_{\beta}\) be the unique \(\tensor\)-decompositions of \(X_{1} \tensor \dots \tensor X_{m} = Z_{1} \tensor \dots \tensor Z_{n}\) and \(Y_{1} \tensor \dots \tensor Y_{m} = W_{1} \tensor \dots \tensor W_{n}\), respectively.
  Then, there are \(x \leq \alpha\) and \(y \leq \beta\) such that \(A_{1} \tensor \dots \tensor A_{x} = X_{1} \tensor \dots \tensor X_{m-1}\) and \(B_{1} \tensor \dots \tensor B_{y} = Y_{1} \tensor \dots \tensor Y_{m-1}\).
  Then, we can rewrite \(\bar{f}\) in terms of \(g_{i}\)s:
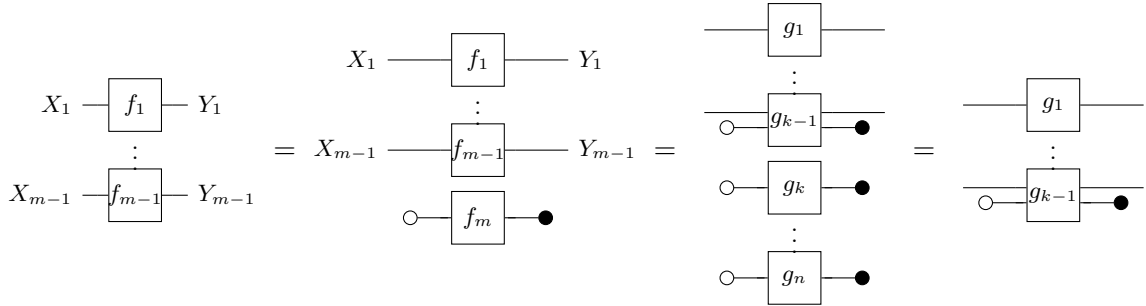
  \[\uniquetensordecProofOne{}\]
  By induction hypothesis, \(\bar{f}\) has a unique \(\tensor\)-decomposition, thus it must be that \(k = m-1\), for every \(i < m-1\) \(f_{i} = g_{i}\) and \(f_{m-1} = (\id{} \tensor \zero) \dcomp g_{k} \dcomp (\id{} \tensor \delete)\).
  Then, we can express \(f_{m}\) in terms of \(g_{m}, \dots, g_{n}\):
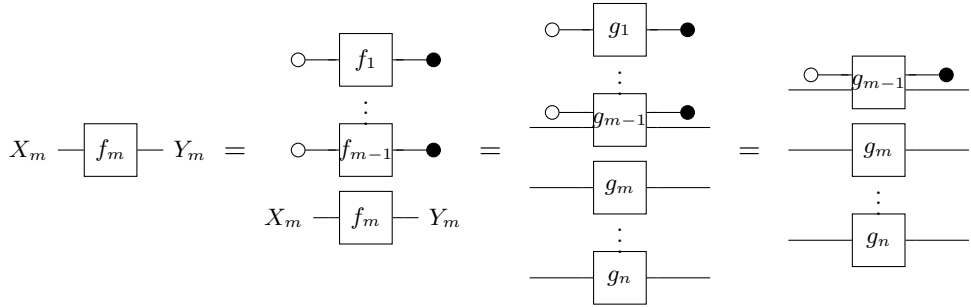
  \[\uniquetensordecProofTwo{}\]
  By hypothesis, \(f_{m}\) is not \(\tensor\)-decomposable and \(m \leq n\).
  Thus, \(n = m\), \(f_{m-1} = g_{m-1}\) and \(f_{m} = g_{m}\).
\end{proof}

Our main result in this section follows from~\Cref{cor:mwd-matrices-blocks} and~\Cref{prop:mon-cat-unique-tensor-decomposition}, which can be applied to \(\Bialg\) because \(0\) is both terminal and initial, and the objects, being a free monoid, are a unique factorisation monoid.

\begin{thm}\label{th:mwd-matrices}
  Let \(f = f_1 \tensor \ldots \tensor f_k\) be a morphism in \(\Bialg\) and its unique \(\tensor\)-decomposition given by \Cref{prop:mon-cat-unique-tensor-decomposition}, with \(r_i = \rank(\BialgToMat(f_i))\).
  Then \(\max\{r_1,\ldots,r_k\} \leq \mwd(f) \leq \max\{r_1,\ldots,r_k\}+1\).
\end{thm}
Note that the identity matrix has monoidal width \(1\)
and twice the identity matrix has monoidal width \(2\), attaining both the upper and lower bounds for the monoidal width of a matrix.

\section{A monoidal algebra for rank width}\label{sec:mwd-rwd}
After having studied monoidal width in the monoidal category of matrices, we are ready to introduce the second monoidal category of \virgolette{open graphs}, which relies on matrices to encode the connectivity of graphs.
In this setting, we capture rank width: we show that instantiating monoidal width in this monoidal category of graphs is equivalent to rank width.

After recalling rank width in \Cref{sec:rwd}, we define the intermediate notion of inductive rank decomposition in \Cref{sec:rec-rank-dec}, and show its equivalence to that of rank decomposition.
As for branch decompositions, adding this intermediate step allows a clearer presentation of the correspondence between rank decompositions and monoidal decompositions.
\Cref{sec:prop-graph} recalls the categorical algebra of graphs with boundaries~\cite{chantawibul2015compositionalgraphtheory,networkGamesCSL}.
Finally, \Cref{sec:mwd-rwd-proof} contains the main result of the present section, which relates inductive rank decompositions, and thus rank decompositions, with monoidal decompositions.

Rank decompositions were originally defined for undirected graphs~\cite{oum2006rank-width}.
This motivates us to consider graphs rather than hypergraphs as in \Cref{sec:mwd-twd}.
As mentioned in \Cref{def:hyperedge-size}, a finite undirected graph is a finite undirected hypergraph with hyperedge size \(2\).
More explicitly,

\begin{defi}\label{def:adjacency-matrix}
  A \emph{graph} \(G = \mathgraph{E}{V}\) is given by a finite set of vertices \(V\), a finite set of edges \(E\) and an adjacency function \(\edgeendsfun \colon E \to \parti_{\leq 2}(V)\), where \(\parti_{\leq 2}(V)\) indicates the set of subsets of \(V\) with at most two elements.
  The same information recorded in the function \(\edgeendsfun\) can be encoded in an equivalence class of matrices, an \emph{adjacency matrix} \(\adjeqclass{G}\):
  the sum of the entries \((i,j)\) and \((j,i)\) of this matrix records the number of edges between vertex \(i\) and vertex \(j\); two adjacency matrices are equivalent when they encode the same graph, i.e. \(\adjeqclass{G} = \adjeqclass{H}\) iff \(G + \transpose{G} = H + \transpose{H}\).
\end{defi}

\subsection{Background: rank width}\label{sec:rwd}
Intuitively, rank width measures the amount of information needed to construct a graph by adding edges to a discrete graph.
Constructing a clique requires little information: we add an edge between any two vertices.
This is reflected in the fact that cliques have rank width \(1\).

Rank width relies on rank decompositions.
In analogy with branch decompositions, a rank decomposition records in a tree a way of iteratively partitioning the vertices of a graph.

\begin{defiC}[\cite{oum2006rank-width}]
  A \emph{rank decomposition} \((Y,r)\) of a graph \(G\) is given by a subcubic tree \(Y\) together with a bijection \(r \colon \leaves(Y) \to \vertices(G)\).
\end{defiC}

Each edge \(b\) in the tree \(Y\) determines a splitting of the graph: it determines a two partition of the leaves of \(Y\), which, through \(r\), determines a 2-partition \(\{A_b,B_b\}\) of the vertices of \(G\).
This corresponds to a splitting of the graph \(G\) into two subgraphs \(G_{1}\) and \(G_{2}\).
Intuitively, the order of an edge \(b\) is the amount of information required to recover \(G\) by joining \(G_{1}\) and \(G_{2}\).
Given the partition \(\{A_b,B_b\}\) of the vertices of \(G\), we can record the edges in \(G\) beween \(A_{b}\) and \(B_{b}\) in a matrix \(X_{b}\).
This means that, if \(v_{i} \in A_{b}\) and \(v_{j} \in B_{b}\), the entry \((i,j)\) of the matrix \(X_{b}\) is the number of edges between \(v_{i}\) and \(v_{j}\).

\begin{defi}[Order of an edge]\label{def:edge-order}
  Let \((Y,r)\) be a rank decomposition of a graph \(G\).
  Let \(b\) be an edge of \(Y\).
  The order of \(b\) is the rank of the matrix associated to it: \(\edgeorder(b) \defn \rank(X_{b})\).
\end{defi}

Note that the order of the two sets in the partition does not matter as the rank is invariant to transposition.
The width of a rank decomposition is the maximum order of the edges of the tree and the rank width of a graph is the width of its cheapest decomposition.

\begin{defi}[Rank width]
  Given a rank decomposition \((Y,r)\) of a graph \(G\), define its width as \(\decwidth(Y,r) \defn \max_{b \in \edges(Y)} \edgeorder(b)\).
  The \emph{rank width} of \(G\) is given by the min-max formula:
  \[\rankwidth(G) \defn \min_{(Y,r)} \decwidth(Y,r).\]
\end{defi}

\subsection{Graphs with dangling edges and inductive definition}\label{sec:rec-rank-dec}
We introduce graphs with dangling edges and inductive rank decomposition of them.
These decompositions are an intermediate notion between rank decompositions and monoidal decompositions.
Similarly to the definition of inductive branch decomposition (\Cref{sec:rec-branch-dec}), they add to rank decompositions the algebraic flavour of monoidal decompositions by using the inductive data type of binary trees to encode a decomposition.

Intuitively, a graph with dangling edges is a graph equipped with some extra edges that connect some vertices in the graph to some boundary ports.
This allows us to combine graphs with dangling edges by connecting some of their dangling edges.
Thus, the equivalence between rank decompositions and inductive rank decompositions formalises the intuition that a rank decomposition encodes a way of dividing a graph into smaller subgraphs by \virgolette{cutting} along some edges.

\begin{defi}\label{def:dangling-graph}
A \emph{graph with dangling edges} \(\Gamma = \danglinggraph{G}{B}\) is given by an adjacency matrix \(G \in \MatN(k,k)\) that records the connectivity of the graph and a matrix \(B \in \MatN(k,n)\) that records the ``dangling edges'' connected to \(n\) boundary ports.
We will sometimes write \(G \in \adjacency(\Gamma)\) and \(B = \boundary(\Gamma)\).
\end{defi}

\begin{exa}\label{ex:graph-dangling-edges-glueing}
  Two graphs with the same ports, as illustrated below, can be ``glued'' together:
  \[\danglinggraphIdeaExFigOne{} \quad\text{ glued with }\quad \danglinggraphIdeaExFigTwo{} \quad\text{ gives }\quad \danglinggraphIdeaExFigThree{}\]
\end{exa}

A rank decomposition is, intuitively, a recipe for decomposing a graph into its single-vertex subgraphs by cutting along its edges.
The cost of each cut is given by the rank of the adjacency matrix that represents it.
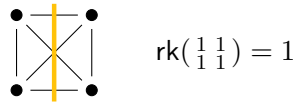
\[\cutrankExFig{}\]

Decompositions are elements of a tree data type, with nodes carrying subgraphs \(\Gamma'\) of the ambient graph \(\Gamma\).
In the following \(\Gamma'\) ranges over the non-empty subgraphs of \(\Gamma\):
\(T_{\Gamma} \ \Coloneqq \ \leafgenerator{\Gamma'} \ \mid \ \nodegenerator{T_{\Gamma}}{\Gamma'}{T_{\Gamma}}\).
Given $T\in T_\Gamma$, the label function $\labelling$ takes a decomposition and returns the graph with dangling edges at the root: $\labelling\nodegenerator{T_1}{\Gamma}{T_2} \defn \Gamma$ and \(\labelling(\leafgenerator{\Gamma}) \defn \Gamma\).

The conditions in the definition of inductive rank decomposition ensure that, by glueing \(\Gamma_1\) and \(\Gamma_2\) together, we get \(\Gamma\) back.

\begin{defi}\label{def:rec-rank-dec}
  Let \(\Gamma = \danglinggraph{G}{B}\) be a graph with dangling edges, where \(G \in  \MatN(k,k)\) and \(B \in \MatN(k,n)\).
  An \emph{inductive rank decomposition} of \(\Gamma\) is \(T \in T_{\Gamma}\) where either:
  \(\Gamma\) is empty and \(T = \emptydec\);
  or \(\Gamma\) has one vertex and \(T = \leafgenerator{\Gamma}\);
  or \(T = \nodegenerator{T_1}{\Gamma}{T_2}\) and \(T_i \in T_{\Gamma_i}\) are inductive rank decompositions of subgraphs \(\Gamma_i = \danglinggraph{G_i}{B_i}\) of \(\Gamma\) such that:
  \begin{itemize}
    \item The vertices are partitioned in two, \(\adjeqclass{G} = \adjeqclass{\begin{psmallmatrix} G_1 & C \\ \zeromat & G_2 \end{psmallmatrix}}\);
    \item The dangling edges are those to the original boundary and to the other subgraph, \(B_1 = (A_1 \mid C)\) and \(B_2 = (A_2 \mid \transpose{C})\), where \(B = \begin{psmallmatrix} A_1 \\ A_2 \end{psmallmatrix}\).
  \end{itemize}
\end{defi}

We will sometimes write \(\Gamma_i = \labelling(T_i)\), \(G_i = \adjacency(\Gamma_i)\) and \(B_i = \boundary(\Gamma_i)\).
We can always assume that the rows of \(G\) and \(B\) are ordered like the leaves of \(T\) so that we can actually split \(B\) horizontally to get \(A_1\) and \(A_2\).

\begin{rem}\label{rem:rwd-vs-twd}
  The perspective on rank width and branch width given by their inductive definitions emphasises an operational difference between them:
  a branch decompositon gives a recipe to construct a graph from its one-edge subgraphs by identifying some of their vertices;
  on the other hand, a rank decomposition gives a recipe to construct a graph from its one-vertex components by connecting some of their \virgolette{dangling} edges.
\end{rem}

\begin{defi}
  Let \(T = \nodegenerator{T_1}{\Gamma}{T_2}\) be an inductive rank decomposition of \(\Gamma = \danglinggraph{G}{B}\), with \(T_i\) possibly both empty.
  Define the \emph{width} of \(T\) inductively: if \(T\) is empty,
  \(\decwidth(\emptydec) \defn 0\);
  otherwise, \(\decwidth(T) \defn \max\{\decwidth(T_1), \decwidth(T_2), \rank(B)\}\).
  Expanding this expression, we obtain
  \[\decwidth(T) = \max_{T' \text{ full subtree of } T} \rank(\boundary(\labelling(T'))).\]
  The \emph{inductive rank width} of \(\Gamma\) is defined by the min-max formula  \(\rrankwidth(\Gamma) \defn \min_{T} \decwidth(T)\).
\end{defi}

We show that the inductive rank width of \(\Gamma = \danglinggraph{G}{B}\) is the same as the rank width of \(G\), up to the rank of the boundary matrix \(B\).

Before proving the upper bound for inductive rank width, we need a technical lemma that relates the width of a graph with that of its subgraphs and allows us to compute it \virgolette{globally}.

\begin{lem}\label{lemma:closed-expr-boundary}
  Let \(T\) be an inductive rank decomposition of \(\Gamma = \danglinggraph{G}{B}\).
  Let \(T'\) be a full subtree of \(T\) and \(\Gamma' \defn \labelling(T')\) with \(\Gamma'= \danglinggraph{G'}{B'}\).
  The adjacency matrix of \(\Gamma\) can be written as \(\adjeqclass{G} = \adjeqclass{\begin{psmallmatrix} G_L & C_L & C \\ \zeromat & G' & C_R \\ \zeromat & \zeromat & G_R \end{psmallmatrix}}\) and its boundary as \(B = \begin{psmallmatrix} A_L \\ A' \\ A_R \end{psmallmatrix}\).
  Then, \(\rank(B') = \rank(A' \mid \transpose{C_L} \mid C_R)\).
\end{lem}
\begin{proof}
  Proceed by induction on the decomposition tree \(T\).
  If it is just a leaf, \(T = \leafgenerator{\Gamma}\), then \(\Gamma\) has at most one vertex, and \(\Gamma' = \emptyset\) or \(\Gamma' = \Gamma\).
  In both cases, the desired equality is true.

  If \(T = \nodegenerator{T_1}{\Gamma}{T_2}\), then, by the definition of inductive rank decomposition, \(\labelling(T_i) = \Gamma_i = \danglinggraph{G_i}{B_i}\) with \(\adjeqclass{G} = \adjeqclass{\begin{psmallmatrix} G_1 & C \\ \zeromat & G_2 \end{psmallmatrix}}\), \(B = \begin{psmallmatrix} A_1 \\ A_2 \end{psmallmatrix}\), \(B_1 = (A_1 \mid C)\) and \(B_2 = (A_2 \mid \transpose{C})\).
  Suppose that \(T' \subtreeq T_1\).
  Then, we can write \(\adjeqclass{G_1} = \adjeqclass{\begin{psmallmatrix} G_L & C_L & D' \\ \zeromat & G' & D_R \\ \zeromat & \zeromat & H_R \end{psmallmatrix}}\), \(A_1 = \begin{psmallmatrix}A_L \\ A' \\ F_R\end{psmallmatrix}\) and \(C = \begin{psmallmatrix}E_L \\ E' \\ E_R\end{psmallmatrix}\).
  It follows that \(B_1 = \begin{psmallmatrix}A_L & E_L \\ A' & E' \\ F_R & E_R\end{psmallmatrix}\) and \(C_R = (D_R \mid E')\).
  By induction hypothesis, \(\rank(B') = \rank(A' \mid E' \mid \transpose{C_L} \mid D_R)\).
  The rank is invariant to permuting the order of columns, thus \(\rank(B') = \rank(A' \mid \transpose{C_L} \mid D_R \mid E') = \rank(A' \mid \transpose{C_L} \mid C_R)\).
  We proceed analogously if \(T' \subtreeq T_2\).
\end{proof}

The above result allows us to relate the width of rank decompositions, which is computed \virgolette{globally}, to the width of inductive rank decompositions, which is computed \virgolette{locally}, with the following bound.

\begin{prop}\label{prop:recursive-rwd-upper-bound}
  Let \(\Gamma = \danglinggraph{G}{B}\) be a graph with dangling edges and \((Y,r)\) be a rank decomposition of \(G\).
  Then, there is an inductive rank decomposition \(\toRecursiveDec(Y,r)\) of \(\Gamma\) such that \(\decwidth(\toRecursiveDec(Y,r)) \leq \decwidth(Y,r) + \rank(B)\).
\end{prop}
\begin{proof}
  Proceed by induction on the number of edges of the decomposition tree \(Y\) to construct an inductive decomposition tree \(T\) in which every non-trivial full subtree \(T'\) has a corresponding edge \(b'\) in the tree \(Y\).
  Suppose \(Y\) has no edges, then either \(G = \emptyset\) or \(G\) has one vertex.
  In either case, we define an inductive rank decomposition with just a leaf labelled with \(\Gamma\), \(\toRecursiveDec(Y,r) \defn \leafgenerator{\Gamma}\).
  We compute its width by definition: \(\decwidth(\toRecursiveDec(Y,r)) \defn \rank(B) \leq \decwidth(Y,r) + \rank(B)\).

  If the decomposition tree has at least an edge, then it is composed of two subcubic subtrees, \(Y = Y_1 \overset{b}{\text{---}} Y_2\).
  Let \(V_i \defn r(\leaves(Y_i))\) be the set of vertices associated to \(Y_{i}\) and \(G_i \defn \subgraphon{G}{V_i}\) be the subgraph of \(G\) induced by the set of vertices \(V_i\).
  By induction hypothesis, there are inductive rank decompositions \(T_i\) of \(\Gamma_i = \danglinggraph{G_i}{B_i}\) in which every full subtree \(T'\) has an associated edge \(b'\).
  Associate the edge \(b\) to both \(T_{1}\) and \(T_{2}\) so that every subtree of \(T\) has an associated edge in \(Y\).
  We can use these decompositions to define an inductive rank decomposition \(T = \nodegenerator{T_1}{\Gamma}{T_2}\) of \(\Gamma\).
  Let \(T'\) be a full subtree of \(T\) corresponding to \(\Gamma' = \danglinggraph{G'}{B'}\).
  By \Cref{lemma:closed-expr-boundary}, we can compute the rank of its boundary matrix \(\rank(B') = \rank(A' \mid \transpose{C_{L}} \mid C_{R})\), where \(A'\), \(C_{L}\) and \(C_{R}\) are defined as in the statement of \Cref{lemma:closed-expr-boundary}.
  The matrix \(A'\) contains some of the rows of \(B\), then its rank is bounded by the rank of \(B\) and we obtain \(\rank(B') \leq \rank(B) + \rank(\transpose{C_{L}} \mid C_{R})\).
  The matrix \((\transpose{C_{L}} \mid C_{R})\) records the edges between the vertices in \(G'\) and the vertices in the rest of \(G\), which, by definition, are the edges that determine \(\edgeorder(b')\).
  This means that the rank of this matrix is the order of the edge \(b'\): \(\rank(\transpose{C_{L}} \mid C_{R}) = \edgeorder(b')\).
  With these observations, we can compute the width of \(T\).
  \begin{align*}
    & \decwidth(T) \\
    & = \max_{T' \subtreeq T} \rank(B')\\
    & = \max_{T' \subtreeq T} \rank(A'\mid \transpose{C_{L}} \mid C_{R})\\
    & \leq \max_{T' \subtreeq T} \rank(\transpose{C_{L}} \mid C_{R}) + \rank(B)\\
    & = \max_{b \in \edges(Y)} \edgeorder(b) + \rank(B)\\
    & \codefn \decwidth(Y,r) + \rank(B)
  \qedhere
  \end{align*}
\end{proof}

\begin{prop}\label{prop:recursive-rwd-lower-bound}
  Let \(T\) be an inductive rank decomposition of \(\Gamma = \danglinggraph{G}{B}\) with \(G \in \MatN(k,k)\) and \(B \in \MatN(k,n)\).
  Then, there is a rank decomposition \(\fromRecursiveDec(T)\) of \(G\) such that \(\decwidth(\fromRecursiveDec(T)) \leq \decwidth(T)\).
\end{prop}
\begin{proof}
  A binary tree is, in particular, a subcubic tree.
  Then, the rank decomposition corresponding to an inductive rank decomposition \(T\) can be defined by its underlying unlabelled tree \(Y\).
  The corresponding bijection \(r \colon \leaves(Y) \to \vertices(G)\) between the leaves of \(Y\) and the vertices of \(G\) can be defined by the labels of the leaves in \(T\):
  the label of a leaf \(l\) of \(T\) is a subgraph of \(\Gamma\) with one vertex \(v_l\) and these subgraphs need to give \(\Gamma\) when composed together.
  Then, the leaves of \(T\), which are the leaves of \(Y\), are in bijection with the vertices of \(G\): there is a bijection \(r \colon \leaves(Y) \to \vertices(G)\) such that \(r(l) \defn v_l\).
  Then, \((Y,r)\) is a branch decomposition of \(G\) and we can define \(\fromRecursiveDec(T) \defn (Y,r)\).

  By construction, the edges of \(Y\) are the same as the edges of \(T\) so we can compute the order of the edges in \(Y\) from the labellings of the nodes in \(T\).
  Consider an edge \(b\) in \(Y\) and consider its endpoints in \(T\): let \(\{v,v_b\} = \edgeends{b}\) with \(v\) parent of \(v_b\) in \(T\).
  The order of \(b\) is related to the rank of the boundary of the subtree \(T_b\) of \(T\) with root in \(v_b\).
  Let \(\labelling(T_b) = \Gamma_b = \danglinggraph{G_b}{B_b}\) be the subgraph of \(\Gamma\) identified by \(T_{b}\).
  We can express the adjacency and boundary matrices of \(\Gamma\) in terms of those of \(\Gamma_{b}\):
  \[\adjeqclass{G} = \adjeqclass{\begin{psmallmatrix} G_L & C_L & C \\ \zeromat & G_b & C_R \\ \zeromat & \zeromat & G_R \end{psmallmatrix}} \qquad \text{and} \qquad B = \begin{psmallmatrix}A_L \\ A' \\ A_R\end{psmallmatrix}.\]
  By \Cref{lemma:closed-expr-boundary}, the boundary rank of \(\Gamma_{b}\) can be computed by \(\rank(B_b) = \rank(A' \mid \transpose{C_L} \mid C_R)\).
  By definition, the order of the edge \(b\) is \(\edgeorder(b) \defn \rank(\transpose{C_L} \mid C_R)\), and we can bound it with the boundary rank of \(\Gamma_{b}\): \(\rank(B_b) \geq \edgeorder(b)\).
  These observations allow us to bound the width of the rank decomposition \(Y\) that corresponds to \(T\).
  \begin{align*}
    & \decwidth(Y,r) \\
    & \defn \max_{b \in \edges(Y)} \edgeorder(b)\\
    & \leq \max_{b \in \edges(Y)} \rank(B_b) \\
    & \leq \max_{T' \leq T} \rank(\boundary(\labelling(T')))\\
    & \codefn \decwidth(T)
    \qedhere
  \end{align*}
\end{proof}

Combining \Cref{prop:recursive-rwd-upper-bound} and \Cref{prop:recursive-rwd-lower-bound} we obtain:

\begin{prop}\label{prop:rec-rank-width-equivalent}
  Inductive rank width is equivalent to rank width.
\end{prop}

\subsection{A prop of graphs}\label{sec:prop-graph}
Here we recall the algebra of graphs with boundaries and its diagrammatic syntax~\cite{networkGamesCSL}.
Graphs with boundaries are graphs together with some extra \virgolette{dangling} edges that connect the graph to the left and right boundaries.
They compose by connecting edges that share a common boundary. All the information about connectivity is handled with matrices.

\begin{rem}
  The categorical algebra of graphs with boundaries is a natural choice for capturing rank width because it emphasises the operation of splitting a graph into parts that share some \emph{edges}.
  This contrasts with the algebra of cospans of graphs (\Cref{sec:cospans-hypergraphs}), in which graphs are split into subgraphs that share some \emph{vertices}.
  The difference in the operation that is emphasised by these two algebras reflects the difference between rank width and tree or branch width pointed out in \Cref{rem:rwd-vs-twd}.
\end{rem}

\begin{defiC}[\cite{networkGamesCSL}]
  A \emph{graph with boundaries} \(g \colon n \to m\) is a tuple $g =$\linebreak[4]$\fullboundariesgraph{G}{L}{R}{P}{F}$ of an adjacency matrix \(\adjeqclass{G}\) of a graph on \(k\) vertices, with \(G \in \MatN(k,k)\); matrices \(L \in \MatN(k,n)\) and \(R \in \MatN(k,m)\) that record the connectivity of the vertices with the left and right boundary; a matrix \(P \in \MatN(m,n)\) that records the passing wires from the left boundary to the right one; and a matrix \(F \in \MatN(m,m)\) that records the wires from the right boundary to itself.
  Graphs with boundaries are taken up to an equivalence making the order of the vertices immaterial.
  Let \(g, g' \colon n \to m\) on \(k\) vertices, with \(g = \fullboundariesgraph{G}{L}{R}{P}{F}\) and \(g' = \fullboundariesgraph{G'}{L'}{R'}{P}{F}\).
  The graphs \(g\) and \(g'\) are considered equal iff
  there is a permutation matrix \(\sigma \in \MatN(k,k)\) such that
  \(g' = \fullboundariesgraph{\sigma G \transpose{\sigma}}{\sigma L}{\sigma R}{P}{F}\).
\end{defiC}
Graphs with boundaries can be composed sequentially and in parallel~\cite{networkGamesCSL}, forming a symmetric monoidal category \(\matGraph\).

The prop \(\propGraph\) provides a convenient syntax for graphs with boundaries.
It is obtained by adding a cup and a vertex generators to the prop of matrices \(\Bialg\) (\Cref{fig:bialgebra}).
\begin{defiC}[\cite{chantawibul2015compositionalgraphtheory}]
  The prop of graphs \(\propGraph\) is obtained by adding to \(\Bialg\) the generators \(\cup \colon 0 \to 2\) and \(\vertex \colon 1 \to 0\) with the equations below.
  \[\propGraphEqFig{}\]
\end{defiC}
These equations mean, in particular, that the cup transposes matrices
(\Cref{fig:adding-cup}, left) and that we can express the equivalence relation of adjacency matrices as in \Cref{def:adjacency-matrix}: \(\adjeqclass{G} = \adjeqclass{H}\) iff \(G + \transpose{G} = H + \transpose{H}\) (\Cref{fig:adding-cup}, right).
\begin{figure}[h!]
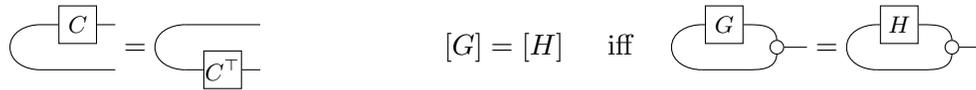

  \centering
  \(\cupTransposeFig{}\) \qquad\qquad\qquad\(\adjeqrelFig{}\)
  \caption{Adding the cup.}\label{fig:adding-cup}
\end{figure}

\begin{propC}[{\cite[Theorem 23]{networkGamesCSL}}]\label{prop:normal-form-prop-graphs}
  The prop of graphs \(\propGraph\) is isomorphic to the prop \(\matGraph\).
\end{propC}
\Cref{prop:normal-form-prop-graphs} means that the morphisms in \(\propGraph\) can be written in the following normal form
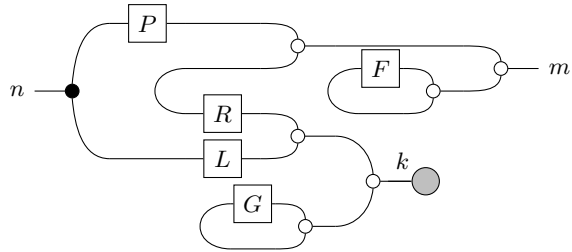
\[\morphismpropgraphExFig{}.\]

The prop \(\propGraph\) is more expressive than graphs with dangling edges (\Cref{def:dangling-graph}): its morphisms can have edges between the boundaries as well.
In fact, graphs with dangling edges can be seen as morphisms \(n \to 0\) in \(\propGraph\).

\begin{exa}
  A graph with dangling edges \(\Gamma = \danglinggraph{G}{B}\) can be represented as a morphism in \(\propGraph\)
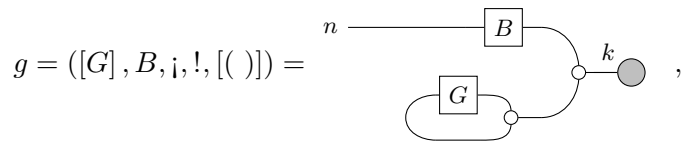
  \[g = \fullboundariesgraph{G}{B}{\initmap{}}{\finmap{}}{\emptymat} = \danglinggraphExFig{}\quad,\]
  where \(\finmap{} \colon n \to 0\) and \(\initmap{} \colon 0 \to k\) are the unique maps to and from the terminal and initial object \(0\).
  We can now formalise the intuition of glueing graphs with dangling edges as explained in \Cref{ex:graph-dangling-edges-glueing}.
  The two graphs there correspond to \(g_{1}\) and \(g_{2}\) below left and middle.
  Their glueing is obtained by precomposing their monoidal product with a cup, i.e. \ \(\cup_{2} \dcomp (g_{1} \tensor g_{2})\), as shown below right.
  \[\danglinggraphGlueingExFig\]
\end{exa}

\begin{defi}\label{def:weight-prop-graph}
  Let the set of \emph{atomic morphisms} \(\decgenerators\) be the set of all the morphisms of \(\propGraph\).
  The \emph{weight function} \(\nodeweight \colon \decgenerators \union \monoidaloperations{\propGraph} \to \naturals\) is defined, on objects \(n\), as \(\nodeweight(\dcompnode{n}) \defn n\); and, on morphisms \(g \in \decgenerators\), as \(\nodeweight(g) \defn k\), where \(k\) is the number of vertices of \(g\).
\end{defi}

Note that, the monoidal width of \(g\) is bounded by the number \(k\) of its vertices, thus we could take as atoms all the morphisms with at most one vertex and the results would not change.

\subsection{Rank width as monoidal width}\label{sec:mwd-rwd-proof}
We show that monoidal width in the prop \(\propGraph\), with the weight function given in \Cref{def:weight-prop-graph}, is equivalent to rank width.
We do this by bounding monoidal width by above with twice rank width and by below with half of rank width (\Cref{th:mwd-rwd}).
We prove these bounds by defining maps from inductive rank decompositions to monoidal decompositions that preserve the width (\Cref{prop:mwd-rwd-upper-bound}), and vice versa (\Cref{prop:mwd-rwd-lower-bound}).

The upper bound (\Cref{prop:mwd-rwd-upper-bound}) is established by associating to each inductive rank decomposition a suitable monoidal decomposition.
This mapping is defined inductively, given the inductive nature of both these structures.
Given an inductive rank decomposition of a graph \(\Gamma\), we can construct a decomposition of its corresponding morphism \(g\) as shown by the first equality in \Cref{fig:mwd-rwd-upper-bound}.
\begin{figure}[h!]
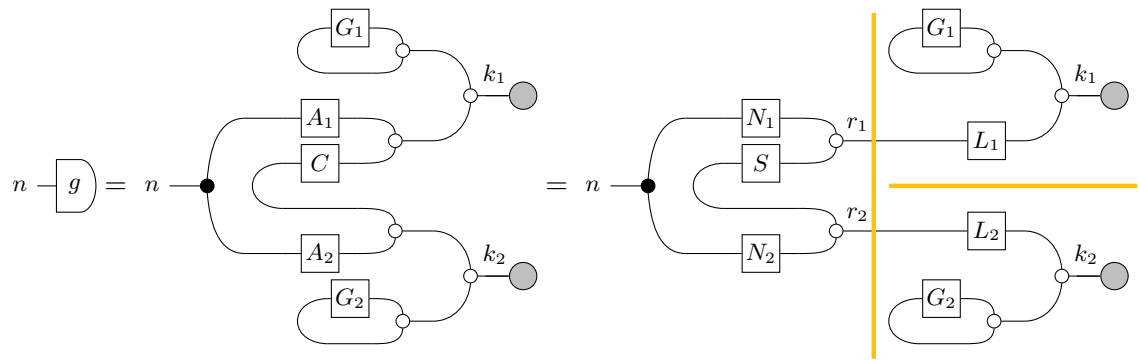

  \[\mwdrwdupperproofFigOneCuts{}\]
  \caption{First step of a monoidal decomposition given by an inductive rank decomposition}\label{fig:mwd-rwd-upper-bound}
\end{figure}
However, this decomposition is not optimal as it cuts along the number of vertices \(k_{1} + k_{2}\).
But we can do better thanks to \Cref{lemma:cut-along-ranks}, which shows that we can cut along the ranks, \(r_1 = \rank(A_1 \mid C)\) and \(r_2 = \rank(A_2 \mid \transpose{C})\), of the boundaries of the induced subgraphs to obtain the second equality in \Cref{fig:mwd-rwd-upper-bound}.
First, recall some facts about ranks.

\begin{rem}\label{rem:rank}
  By \Cref{lemma:min-cut-is-rank}, the rank of a composition of two matrices is bounded by their ranks: \(\rank(A \cdot B) \leq \min\{\rank(A), \rank(B)\}\).
  If, moreover, \(B\) has full rank, then \(\rank(A \cdot B) = \rank(A)\).
\end{rem}

\begin{lem}\label{lemma:cut-along-ranks}
  Let \(A_i \in \MatN(k_i,n)\), for \(i = 1,2\), and \(C \in \MatN(k_1,k_2)\).
  Then, there are rank decompositions of \((A_1 \mid C)\) and \((A_2 \mid \transpose{C})\) of the form \((A_1 \mid C) = L_1 \cdot (N_1 \mid S \cdot \transpose{L_2})\), and \((A_2 \mid \transpose{C}) = L_2 \cdot (N_2 \mid \transpose{S} \cdot \transpose{L_1})\).
  This ensures that we can decompose the diagram below on the left-hand-side as the one on the right-hand-side, where \(r_{1} = \rank(A_1 \mid C)\) and \(r_2 = \rank(A_2 \mid \transpose{C})\).
  \[\lemmacutalongranksIdeaFig{}\]
\end{lem}
\begin{proof}
  Let \(r_1 = \rank(A_1 \mid C)\) and \(r_2 = \rank(A_2 \mid \transpose{C})\).
  We start by factoring \((A_1 \mid C)\) into \(L_1 \cdot (N_1 \mid K_1)\),
  \[\lemmacutalongranksProofFigOne{}\]
  where \(L_1 \in \MatN(k_1,r_1)\), \(N_1 \in \MatN(r_1,n)\) and \(K_1 \in \MatN(r_1,k_2)\).
  Then, we proceed with factoring \((A_2 \mid \transpose{K_1})\) and we show that \(\rank(A_2 \mid \transpose{K_1}) = \rank(A_2 \mid \transpose{C})\).
  Let \(L_2 \cdot (N_2 \mid K_2)\) be a rank factorisation of \((A_2 \mid \transpose{K_1})\),
  \[\lemmacutalongranksProofFigTwo{}\]
  with \(L_2 \in \MatN(k_2,r')\), \(N_2 \in \MatN(r',n)\) and \(K_2 \in \MatN(r',k_1)\).
  We show that \(r' = r_2\).
  By the first factorisation, we obtain that \(C = L_1 \cdot K_1\), and
  \[(A_2 \mid \transpose{C}) = (A_2 \mid \transpose{K_1} \cdot \transpose{L_1}) = (A_2 \mid \transpose{K_1}) \cdot \begin{psmallmatrix} \idmat{} & \zeromat \\ \zeromat & \transpose{L_1}\end{psmallmatrix}.\]
  Then, \(r' = r_2\) because \(L_1\) and, consequently, \(\begin{psmallmatrix} \idmat{} & \zeromat \\ \zeromat & \transpose{L_1}\end{psmallmatrix}\) have full rank.
  By letting \(S = \transpose{K_2}\), we obtain the desired factorisation.
\end{proof}

Once we have performed the cuts in \Cref{fig:mwd-rwd-upper-bound} on the right, we have changed the boundaries of the induced subgraphs.
This means that we cannot apply the inductive hypothesis right away, but we need to transform first the inductive rank decompositions of the old subgraphs into decompositions of the new ones, as shown in \Cref{lemma:rank-on-boundary}.
More explicitly, when \(M\) has full rank, if we have an inductive rank decomposition of \(\Gamma = \danglinggraph{G}{B \cdot M}\), which corresponds to \(g\) below left, we can obtain one of \(\Gamma' = \danglinggraph{G}{B \cdot M'}\), which corresponds to \(g'\) below right, of at most the same width.
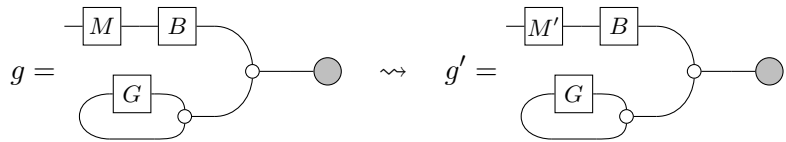
\[\lemmarankonboundaryIdeaFig{}\]

\begin{lem}\label{lemma:rank-on-boundary}
  Let \(T\) be an inductive rank decomposition of \(\Gamma = \danglinggraph{G}{B \cdot M}\), with \(M\) that has full rank.
  Then, there is an inductive rank decomposition \(T'\) of \(\Gamma' = \danglinggraph{G}{B \cdot M'}\) such that \(\decwidth(T) \leq \decwidth(T')\) and such that \(T\) and \(T'\) have the same underlying tree structure.
  If, moreover, \(M'\) has full rank, then \(\decwidth(T) = \decwidth(T')\).
\end{lem}
\begin{proof}
  Proceed by induction on the decomposition tree \(T\).
  If the tree \(T\) is just a leaf with label \(\Gamma\), then we define the corresponding tree to be just a leaf with label \(\Gamma'\): \(T' \defn \leafgenerator{\Gamma'}\).
  Clearly, \(T\) and \(T'\) have the same underlying tree structure.
  By \Cref{rem:rank} and the fact that \(M\) has full rank, we can relate their widths: \(\decwidth(T') \defn \rank(B \cdot M') \leq \rank(B) = \rank(B \cdot M) \codefn \decwidth(T)\).
  If, moreover, \(M'\) has full rank, the inequality becomes an equality and \(\decwidth(T') = \decwidth(T)\).

  If \(T = \nodegenerator{T_1}{\Gamma}{T_2}\), then the adjacency and boundary matrices of \(\Gamma\) can be expressed in terms of those of its subgraphs \(\Gamma_i \defn \labelling_i(T_i) = \danglinggraph{G_{i}}{D_{i}}\), by definition of inductive rank decomposition: \(G = \begin{psmallmatrix} G_1 & C \\ \zeromat & G_2 \end{psmallmatrix}\), \(B \cdot M = \begin{psmallmatrix} A_1 \\ A_2 \end{psmallmatrix} \cdot M =  \begin{psmallmatrix} A_1 \cdot M \\ A_2 \cdot M \end{psmallmatrix}\), with \(D_1 = (A_1 \cdot M \mid C)\) and \(D_2 = (A_2 \cdot M \mid \transpose{C})\).
  The boundary matrices \(D_{i}\) of the subgraphs \(\Gamma_{i}\) can also be expressed as a composition with a full-rank matrix: \(D_1 = (A_1 \cdot M \mid C) = (A_1 \mid C) \cdot \begin{psmallmatrix} M & \zeromat \\ \zeromat & \idmat{k_{2}} \end{psmallmatrix}\) and \(D_2 = (A_2 \cdot M \mid \transpose{C}) = (A_2 \mid \transpose{C}) \cdot \begin{psmallmatrix} M & \zeromat \\ \zeromat & \idmat{k_{1}} \end{psmallmatrix}\).
  The matrices \(\begin{psmallmatrix} M & \zeromat \\ \zeromat & \idmat{k_{i}} \end{psmallmatrix}\) have full rank because all their blocks do.
  Let \(B_{1} = (A_1 \mid C)\) and \(B_{2} = (A_2 \mid \transpose{C})\).
  By induction hypothesis, there are inductive rank decompositions \(T'_1\) and \(T'_{2}\) of \(\Gamma'_1 = \danglinggraph{G_1}{B_1 \cdot \begin{psmallmatrix} M' & \zeromat \\ \zeromat & \idmat{k_{2}} \end{psmallmatrix}}\) and \(\Gamma'_2 = \danglinggraph{G_2}{B_2 \cdot \begin{psmallmatrix} M' & \zeromat \\ \zeromat & \idmat{k_{1}} \end{psmallmatrix}}\) with the same underlying tree structure as \(T_1\) and \(T_{2}\), respectively.
  Moreover, their width is bounded, \(\decwidth(T'_i) \leq \decwidth(T_i)\), and if, additionally, \(M'\) has full rank, \(\decwidth(T'_i) = \decwidth(T_i)\).
  Then, we can use these decompositions to define an inductive rank decomposition \(T' \defn \nodegenerator{T'_1}{\Gamma'}{T'_2}\) of \(\Gamma'\) because its adjacency and boundary matrices can be expressed in terms of those of \(\Gamma'_{i}\) as in the definition of inductive rank decomposition: \(G = \begin{psmallmatrix} G_1 & C \\ \zeromat & G_2 \end{psmallmatrix}\), \(B_1 \cdot \begin{psmallmatrix} M' & \zeromat \\ \zeromat & \idmat{k_{2}} \end{psmallmatrix} = (A_1 \cdot M' \mid C)\) and \(B_2 \cdot \begin{psmallmatrix} M' & \zeromat \\ \zeromat & \idmat{k_{1}} \end{psmallmatrix} = (A_2 \cdot M' \mid \transpose{C})\).
  Applying the induction hypothesis and \Cref{rem:rank}, we compute the width of this decomposition.
  \begin{align*}
    & \decwidth(T')\\
    & \defn \max \{\rank(B \cdot M'), \decwidth(T'_1), \decwidth(T'_2)\}\\
    & \leq \max \{\rank(B), \decwidth(T_1), \decwidth(T_2)\}\\
    & = \max \{\rank(B \cdot M), \decwidth(T_1), \decwidth(T_2)\}\\
    & \codefn \decwidth(T)
  \end{align*}
  If, moreover, \(M'\) has full rank, the inequality becomes an equality and \(\decwidth(T') = \decwidth(T)\).
\end{proof}

With the above ingredients, we can show that rank width bounds monoidal width from above.

\begin{prop}\label{prop:mwd-rwd-upper-bound}
  Let \(\Gamma = \danglinggraph{G}{B}\) be a graph with dangling edges and \(g \colon n \to 0\) be the morphism in \(\propGraph\) corresponding to \(\Gamma\).
  Let \(T\) be an inductive rank decomposition of \(\Gamma\).
  Then, there is a monoidal decomposition \(\rTomdec(T)\) of \(g\) such that \(\decwidth(\rTomdec(T)) \leq 2 \cdot \decwidth(T)\).
\end{prop}
\begin{proof}
  Proceed by induction on the decomposition tree \(T\).
  If it is empty, then \(G\) must also be empty, \(\rTomdec(T) = \emptydec\) and we are done.
  If the decomposition tree consists of just one leaf with label \(\Gamma\), then \(\Gamma\) must have one vertex, we can define \(\rTomdec(T) \defn \leafgenerator{g}\) to also be just a leaf, and bound its width \(\decwidth(T) \defn \rank(G) = \decwidth(\rTomdec(T))\).

  If \(T = \nodegenerator{T_1}{\Gamma}{T_2}\), then we can relate the adjacency and boundary matrices of \(\Gamma\) to those of \(\Gamma_i \defn \labelling(T_i) = \danglinggraph{G_i}{B_i}\), by definition of inductive rank decomposition: \(G = \begin{psmallmatrix} G_1 & C \\ \zeromat & G_2 \end{psmallmatrix}\), \(B = \begin{psmallmatrix} A_1 \\ A_2 \end{psmallmatrix}\), \(B_1 = (A_1 \mid C)\) and \(B_2 = (A_2 \mid \transpose{C})\).
  By \Cref{lemma:cut-along-ranks}, there are rank decompositions of \((A_1 \mid C)\) and \((A_2 \mid \transpose{C})\) of the form: \((A_1 \mid C) = L_1 \cdot (N_1 \mid S \cdot \transpose{L_2})\); and \((A_2 \mid \transpose{C}) = L_2 \cdot (N_2 \mid \transpose{S} \cdot \transpose{L_1})\).
  This means that we can write \(g\) as
  \[\mwdrwdupperproofFigOneCuts{},\]
  with \(r_i = \rank(B_i)\).
  Then, \(B_i = L_i \cdot M_i\) with \(M_i\) that has full rank \(r_i\).
  By taking \(M' = \idmat{}\) in \Cref{lemma:rank-on-boundary}, there is an inductive rank decomposition \(T'_i\) of \(\Gamma'_i = \danglinggraph{G_i}{L_i}\), with the same underlying binary tree as \(T_i\), such that \(\decwidth(T_i) = \decwidth(T'_i)\).
  Let \(g_i \colon r_i \to 0\) be the morphisms in \(\propGraph\) corresponding to \(\Gamma'_i\) and let \(b \colon n \to r_1 + r_2\) be defined as
  \[\mwdrwdupperproofFigTwo{}.\]
  By induction hypothesis, there are monoidal decompositions \(\rTomdec(T'_1)\) and \(\rTomdec(T'_2)\) of \(g_1\) and \(g_{2}\) of bounded width: \(\decwidth(\rTomdec(T'_i)) \leq 2 \cdot \decwidth(T'_i) = 2 \cdot \decwidth(T_i)\).
  Then, \(g = b \dcomp_{r_1 + r_2} (g_1 \tensor g_2)\) and \(\rTomdec(T) \defn \nodegenerator{b}{\dcomp_{r_1 + r_2}}{\nodegenerator{\rTomdec(T'_1)}{\tensor}{\rTomdec(T'_2)}}\) is a monoidal decomposition of \(g\).
  Its width can be computed.
  \begin{align*}
    & \decwidth(\rTomdec(T))\\
    & \defn \max \{\nodeweight(b), \nodeweight(\dcompnode{r_1 + r_2}), \decwidth(\rTomdec(T'_1)), \decwidth(\rTomdec(T'_2))\}\\
    & \leq \max \{\nodeweight(b), \nodeweight(\dcompnode{r_1 + r_2}), 2 \cdot \decwidth(T'_1), 2 \cdot \decwidth(T'_2)\}\\
    & = \max \{\nodeweight(b), r_1 + r_2, 2 \cdot \decwidth(T_1), 2 \cdot \decwidth(T_2)\}\\
    & \leq 2 \cdot \max \{r_1, r_2, \decwidth(T_1), \decwidth(T_2)\}\\
    & \codefn 2 \cdot \decwidth(T)
    \qedhere
  \end{align*}
\end{proof}

Proving the lower bound is similarly involved and follows a similar proof structure.
From a monoidal decomposition we construct inductively an inductive rank decomposition of bounded width.
The inductive step relative to composition nodes is the most involved and needs two additional lemmas, which allow us to transform inductive rank decompositions of the induced subgraphs into ones of two subgraphs that satisfy the conditions of \Cref{def:rec-rank-dec}.

Applying the inductive hypothesis gives us an inductive rank decomposition of \(\Gamma = \danglinggraph{G}{(L \mid R)}\), which is associated to \(g\) below left, and we need to construct one of \(\Gamma' \defn \danglinggraph{G + L \cdot F \cdot \transpose{L}}{(L \mid R + L \cdot (F + \transpose{F}) \cdot \transpose{P})}\), which is associated to \(f \dcomp g\) below right, of at most the same width.
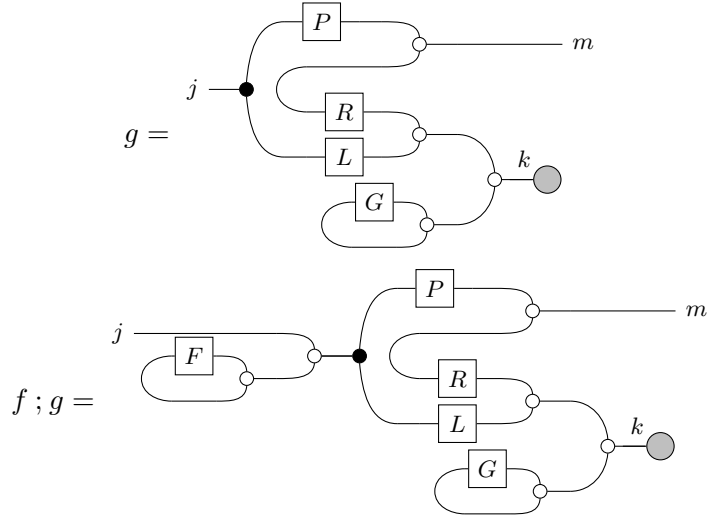
\lemmawirestofutureIdeaFig{}

\begin{lem}\label{lemma:wires-to-future}
  Let \(T\) be an inductive rank decomposition of \(\Gamma = \danglinggraph{G}{(L \mid R)}\), with \(G \in \MatN(k,k)\), \(L \in \MatN(k,j)\) and \(R \in \MatN(k,m)\).
  Let \(F \in \MatN(j,j)\), \(P \in \MatN(m,j)\) and define \(\Gamma' \defn \danglinggraph{G + L \cdot F \cdot \transpose{L}}{(L \mid R + L \cdot (F + \transpose{F}) \cdot \transpose{P})}\).
  Then, there is an inductive rank decomposition \(T'\) of \(\Gamma'\) such that \(\decwidth(T') \leq \decwidth(T)\).
\end{lem}
\begin{proof}
  Note that we can factor the boundary matrix of \(\Gamma'\) as \((L \mid R + L \cdot (F + \transpose{F}) \cdot \transpose{P}) = (L \mid R) \cdot \begin{psmallmatrix} \id{j} & (F + \transpose{F}) \cdot \transpose{P} \\ \zeromat & \id{m}\end{psmallmatrix}\).
  Then, we can bound its rank, \(\rank(L \mid R + L \cdot (F + \transpose{F}) \cdot \transpose{P}) \leq \rank (L \mid R)\).

  Proceed by induction on the decomposition tree \(T\).
  If it is just a leaf with label \(\Gamma\), then \(\Gamma\) has one vertex and we can define a decomposition for \(\Gamma'\) to be also just a leaf: \(T' \defn \leafgenerator{\Gamma'}\).
  We can bound its width with the width of \(T\): \(\decwidth(T') \defn \rank(L \mid R + L \cdot (F + \transpose{F}) \cdot \transpose{P}) \leq \rank(L \mid R) \codefn \decwidth(T)\).

  If \(T = \nodegenerator{T_1}{\Gamma}{T_2}\), then there are two subgraphs \(\Gamma_1 = \danglinggraph{G_1}{(L_1 \mid R_1 \mid C)}\) and \(\Gamma_2 = \danglinggraph{G_2}{(L_2 \mid R_2 \mid C)}\) such that \(T_i\) is an inductive rank decomposition of \(\Gamma_i\), and we can relate the adjacency and boundary matrices of \(\Gamma\) to those of \(\Gamma_{1}\) and \(\Gamma_{2}\), by definition of inductive rank decomposition: \(\adjeqclass{G} = \adjeqclass{\begin{psmallmatrix} G_1 & C \\ \zeromat & G_2 \end{psmallmatrix}}\) and \((L \mid R) = \begin{psmallmatrix} L_1 & R_1 \\ L_2 & R_2 \end{psmallmatrix}\).
  Similarly, we express the adjacency and boundary matrices of \(\Gamma'\) in terms of the same components: \(\adjeqclass{G + L \cdot F \cdot \transpose{L}} = \adjeqclass{\begin{psmallmatrix} G_1 + L_1 \cdot F \cdot \transpose{L_1} & C + L_1 \cdot (F + \transpose{F}) \cdot \transpose{L_2} \\ \zeromat & G_2 + L_2 \cdot F \cdot \transpose{L_2}\end{psmallmatrix}}\) and \((L \mid R + L \cdot (F + \transpose{F}) \cdot \transpose{P}) = \begin{psmallmatrix} L_1 & R_1 + L_1 \cdot (F + \transpose{F}) \cdot \transpose{P} \\ L_2 & R_2 + L_2 \cdot (F + \transpose{F}) \cdot \transpose{P} \end{psmallmatrix}\).
  We use these decompositions to define two subgraphs of \(\Gamma'\) and apply the induction hypothesis to them.
  \begin{align*}
    \Gamma'_1 \defn & \danglinggraph{G_1 + L_1 \cdot F \cdot \transpose{L_1}}{(L_1 \mid R_1 + L_1 \cdot (F + \transpose{F}) \cdot \transpose{P} \mid C + L_1 \cdot (F + \transpose{F}) \cdot \transpose{L_2})} \\
    =& \danglinggraph{G_1 + L_1 \cdot F \cdot \transpose{L_1}}{(L_1 \mid (R_1 \mid C) + L_1 \cdot (F + \transpose{F}) \cdot (\transpose{P} \mid \transpose{L_2}))}\\
    \text{ and }\\
    \Gamma'_2 \defn & \danglinggraph{G_2 + L_2 \cdot F \cdot \transpose{L_2}}{(L_2 \mid R_2 + L_2 \cdot (F + \transpose{F}) \cdot \transpose{P} \mid \transpose{C} + L_2 \cdot (F + \transpose{F}) \cdot \transpose{L_1})} \\
    =& \danglinggraph{G_2 + L_2 \cdot F \cdot \transpose{L_2}}{(L_2 \mid (R_2 \mid \transpose{C}) + L_2 \cdot (F + \transpose{F}) \cdot (\transpose{P} \mid \transpose{L_1}))}
  \end{align*}
  By induction, we have inductive rank decompositions \(T'_i\) of \(\Gamma'_i\) such that \(\decwidth(T'_i) \leq \decwidth(T_i)\).
  We defined \(\Gamma'_i\) so that \(T' \defn \nodegenerator{T'_1}{\Gamma'}{T'_2}\) would be an inductive rank decomposition of \(\Gamma'\).
  We can bound its width as desired.
  \begin{align*}
    & \decwidth(T')\\
    & \defn \max \{\decwidth(T'_1), \decwidth(T'_2), \rank(L \mid R + L \cdot (F + \transpose{F}) \cdot \transpose{P})\}\\
    & \leq \max \{\decwidth(T_1), \decwidth(T_2), \rank(L \mid R + L \cdot (F + \transpose{F}) \cdot \transpose{P})\}\\
    & \leq \max \{\decwidth(T_1), \decwidth(T_2), \rank(L \mid R)\}\\
    & \codefn \decwidth(T)
    \qedhere
  \end{align*}
\end{proof}

In order to obtain the subgraphs of the desired shape we need to add some extra connections to the boundaries.
This can be done thanks to \Cref{lemma:rank-on-boundary}, by taking \(M = \idmat{}\).
We are finally able to prove the lower bound for monoidal width.

\begin{prop}\label{prop:mwd-rwd-lower-bound}
  Let \(g = \fullboundariesgraph{G}{L}{R}{P}{F}\) in \(\propGraph\) and \(d \in \decset{g}\).
  Let \(\Gamma = \danglinggraph{G}{(L \mid R)}\).
  Then, there is an inductive rank decomposition \(\mTordec(d)\) of \(\Gamma\) s.t.
  \(\decwidth(\mTordec(d)) \leq 2 \cdot \max \{\decwidth(d), \rank(L), \rank(R)\}\).
\end{prop}
\begin{proof}
  Proceed by induction on the decomposition tree \(d\).
  If it is just a leaf with label \(g\), then its width is defined to be the number \(k\) of vertices of \(g\), \(\decwidth(d) \defn k\).
  Pick any inductive rank decomposition of \(\Gamma\) and define \(\mTordec(d) \defn T\).
  Surely, \(\decwidth(T) \leq k \codefn \decwidth(d)\)

  If \(d = \nodegenerator{d_1}{\dcomp_j}{d_2}\), then \(g\) is the composition of two morphisms: \(g = g_1 \dcomp g_2\), with \(g_i = \fullboundariesgraph{G_i}{L_i}{R_i}{P_i}{F_i}\).
  Given the partition of the vertices determined by \(g_1\) and \(g_2\), we can decompose \(g\) in another way, by writing \(\adjeqclass{G} = \adjeqclass{\begin{psmallmatrix} \overline{G}_1 & C \\ \zeromat & \overline{G}_2\end{psmallmatrix}}\) and \(B = (L \mid R) = \begin{psmallmatrix} \overline{L}_1 & \overline{R}_1 \\ \overline{L}_2 & \overline{R}_2 \end{psmallmatrix}\).
  Then, we have that \(\overline{G}_1 = G_1\), \(\overline{L}_1 = L_1\), \(P = P_2 \cdot P_1\), \(C = R_1 \cdot \transpose{L}_2\), \(\overline{R}_1 = R_1 \cdot \transpose{P}_2\), \(\overline{L}_2 = L_2 \cdot P_1\), \(\overline{R}_2 = R_2 + L_2 \cdot (F_1 + \transpose{F}_1) \cdot \transpose{P}_2\), \(\overline{G}_2 = G_2 + L_2 \cdot F_1 \cdot \transpose{L}_2\), and \(F = F_2 + P_2 \cdot F_1 \cdot \transpose{P}_2\).
  This corresponds to the following diagrammatic rewriting using the equations of \(\propGraph\).
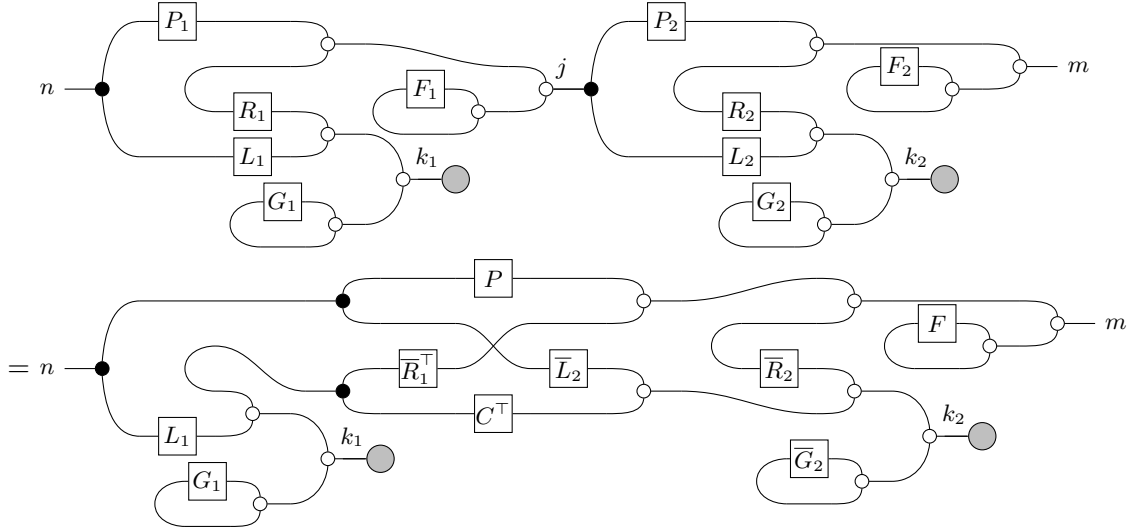
  \[\mwdrankwidthlowerproofFigFuture{}\]
  We define \(\overline{B}_1 \defn (\overline{L}_1 \mid \overline{R}_1 \mid C)\) and \(\overline{B}_2 \defn (\overline{L}_2 \mid \overline{R}_2 \mid \transpose{C})\).
  In order to build an inductive rank decomposition of \(\Gamma\), we need rank decompositions of \(\overline{\Gamma}_i = \danglinggraph{\overline{G}_i}{\overline{B}_i}\).
  We obtain these in three steps.
  Firstly, we apply induction to obtain inductive rank decompositions \(\mTordec(d_i)\) of \(\Gamma_i = \danglinggraph{G_i}{(L_i \mid R_i)}\) such that \(\decwidth(\mTordec(d_i)) \leq 2 \cdot \max \{\decwidth(d_i), \rank(L_i), \rank(R_i)\}\).
  Secondly, we apply \Cref{lemma:wires-to-future} to obtain an inductive rank decomposition \(T'_2\) of \(\Gamma'_2 = \danglinggraph{G_2 + L_2 \cdot F_1 \cdot \transpose{L_2}}{(L_2 \mid R_2 + L_2 \cdot (F_1 + \transpose{F_1}) \cdot \transpose{P_2})}\) such that \(\decwidth(T'_2) \leq \decwidth(\mTordec(d_2))\).
  Lastly, we observe that \((\overline{R}_1 \mid C) = R_1 \cdot (\transpose{P_2} \mid \transpose{L_2})\) and \((\overline{L}_2 \mid \transpose{C}) = L_2 \cdot (P_1 \mid \transpose{R_1})\).
  Then we obtain that \(\overline{B}_1 = (L_1 \mid R_1) \cdot \begin{psmallmatrix} \id{n} & \zeromat & \zeromat\\ \zeromat &  \transpose{P_2} & \transpose{L_2} \end{psmallmatrix}\) and \(\overline{B}_2 = (L_2 \mid R_2 + L_2 \cdot (F_1 + \transpose{F_1}) \cdot \transpose{P_2}) \cdot \begin{psmallmatrix} P_1 & \zeromat & \transpose{R}_1 \\ \zeromat & \id{m} & \zeromat \end{psmallmatrix}\), and we can apply \Cref{lemma:rank-on-boundary}, with \(M = \idmat{}\), to get inductive rank decompositions \(T_i\) of \(\overline{\Gamma}_i\) such that \(\decwidth(T_1) \leq \decwidth(\mTordec(d_1))\) and \(\decwidth(T_2) \leq \decwidth(T'_2) \leq \decwidth(\mTordec(d_2))\).
  If \(k_1, k_2 > 0\), then we define \(\mTordec(d) \defn \nodegenerator{T_1}{\Gamma}{T_2}\), which is an inductive rank decomposition of \(\Gamma\) because \(\overline{\Gamma}_i\) satisfy the two conditions in \Cref{def:rec-rank-dec}.
  If \(k_1 = 0\), then \(\Gamma = \overline{\Gamma}_2\) and we can define \(\mTordec(d) \defn T_2\).
  Similarly, if \(k_2 = 0\), then \(\Gamma = \overline{\Gamma}_1\) and we can define \(\mTordec(d) \defn T_1\).
  In any case, we can compute the width of \(\mTordec(d)\) (if \(k_i = 0\) then \(T_i = \emptydec\) and \(\decwidth(T_i) = 0\)) using the inductive hypothesis, \Cref{lemma:wires-to-future}, \Cref{lemma:rank-on-boundary}, the fact that \(\rank(L) \geq \rank(L_1)\), \(\rank(R) \geq \rank(R_2)\) and \(j \geq \rank(R_1), \rank(L_2)\) because \(R_1 \colon j \to k_1\) and \(L_2 \colon j \to k_2\).
  \begin{align*}
    & \decwidth(T) \\
    & \defn \max \{\decwidth(T_1), \decwidth(T_2), \rank(L \mid R)\}\\
    & \leq \max \{\decwidth(\mTordec(d_1)), \decwidth(T'_2), \rank(L \mid R)\}\\
    & \leq \max \{\decwidth(\mTordec(d_1)), \decwidth(\mTordec(d_2)), \rank(L \mid R)\}\\
    & \leq \max \{\decwidth(\mTordec(d_1)), \decwidth(\mTordec(d_2)), \rank(L) + \rank(R)\}\\
    & \leq \max \{2 \cdot \decwidth(d_1), 2 \cdot \rank(L_1), 2 \cdot \rank(R_1), 2 \cdot \decwidth(d_2), 2 \cdot \rank(L_2), 2 \cdot \rank(R_2), \rank(L) + \rank(R)\}\\
    & \leq 2 \cdot \max \{\decwidth(d_1), \rank(L_1), \rank(R_1), \decwidth(d_2), \rank(L_2), \rank(R_2), \rank(L), \rank(R)\}\\
    & \leq 2 \cdot \max \{\decwidth(d_1), \decwidth(d_2), j, \rank(L), \rank(R)\}\\
    & \codefn 2 \cdot \max \{\decwidth(d), \rank(L), \rank(R)\}
  \end{align*}

  If \(d = \nodegenerator{d_1}{\tensor}{d_2}\), then \(g\) is the monoidal product of two morphisms: \(g = g_1 \tensor g_2\), with \(g_i = \fullboundariesgraph{G_i}{L_i}{R_i}{P_i}{F_i} \colon n_i \to m_i\).
  By exlicitly computing the monoidal product, we obtain that \(\adjeqclass{G} = \adjeqclass{\begin{psmallmatrix} G_1 & \zeromat \\ \zeromat & G_2 \end{psmallmatrix}}\), \(L = \begin{psmallmatrix} L_1 & \zeromat \\ \zeromat & L_2 \end{psmallmatrix}\), \(R = \begin{psmallmatrix} R_1 & \zeromat \\ \zeromat & R_2 \end{psmallmatrix}\), \(P = \begin{psmallmatrix} P_1 & \zeromat \\ \zeromat & P_2 \end{psmallmatrix}\) and \(F = \begin{psmallmatrix} F_1 & \zeromat \\ \zeromat & F_2 \end{psmallmatrix}\).
  By induction, we have inductive rank decompositions \(\mTordec(d_i)\) of \(\Gamma_i \defn \danglinggraph{G_i}{B_i}\), where \(B_i = (L_i \mid R_i)\), of bounded width: \(\decwidth(\mTordec(d_i)) \leq 2 \cdot \max \{\decwidth(d_i), \rank(L_i), \rank(R_i)\}\).
  Let \(\overline{B}_1 \defn (L_1 \mid \zeromat_{n_2} \mid R_1 \mid \zeromat_{m_2} \mid \zeromat_{k_2}) = B_1 \cdot \begin{psmallmatrix} \id{n_1} & \zeromat & \zeromat & \zeromat & \zeromat \\ \zeromat & \zeromat & \id{m_1} & \zeromat & \zeromat \end{psmallmatrix}\) and  \(\overline{B}_2 \defn (\zeromat_{n_1} \mid L_2 \mid \zeromat_{m_1} \mid R_2 \mid \zeromat_{k_1}) = B_2 \cdot \begin{psmallmatrix} \zeromat & \id{n_2} & \zeromat & \zeromat & \zeromat \\ \zeromat & \zeromat & \zeromat & \id{m_2} & \zeromat \end{psmallmatrix}\).
  By taking \(M = \idmat{}\) in \Cref{lemma:rank-on-boundary}, we can obtain inductive rank decompositions \(T_i\) of \(\overline{\Gamma}_i \defn \danglinggraph{G_i}{\overline{B}_i}\) such that \(\decwidth(T_i) \leq \decwidth(\mTordec(d_i))\).
  If \(k_1, k_2 > 0\), then we define \(\mTordec(d) \defn \nodegenerator{T_1}{\Gamma}{T_2}\), which is an inductive rank decomposition of \(\Gamma\) because \(\overline{\Gamma}_i\) satisfy the two conditions in \Cref{def:rec-rank-dec}.
  If \(k_1 = 0\), then \(\Gamma = \overline{\Gamma}_2\) and we can define \(\mTordec(d) \defn T_2\).
  Similarly, if \(k_2 = 0\), then \(\Gamma = \overline{\Gamma}_1\) and we can define \(\mTordec(d) \defn T_1\).
  In any case, we can compute the width of \(\mTordec(d)\) (if \(k_i = 0\) then \(T_i = \emptydec\) and \(\decwidth(T_i) = 0\)) using the inductive hypothesis and \Cref{lemma:rank-on-boundary}.
  \begin{align*}
    & \decwidth(T) \\
    & \defn \max \{\decwidth(T_1), \decwidth(T_2), \rank(L \mid R)\}\\
    & \leq \max \{\decwidth(\mTordec(d_1)), \decwidth(\mTordec(d_2)), \rank(L \mid R)\}\\
    & \leq \max \{\decwidth(\mTordec(d_1)), \decwidth(\mTordec(d_2)), \rank(L) + \rank(R)\}\\
    & \leq \max \{2 \cdot \decwidth(d_1), 2 \cdot \rank(L_1), 2 \cdot \rank(R_1), 2 \cdot \decwidth(d_2), 2 \cdot \rank(L_2), 2 \cdot \rank(R_2), \rank(L) + \rank(R)\}\\
    & \leq 2 \cdot \max \{\decwidth(d_1), \rank(L_1), \rank(R_1), \decwidth(d_2), \rank(L_2), \rank(R_2), \rank(L), \rank(R)\}\\
    & \leq 2 \cdot \max \{\decwidth(d_1), \decwidth(d_2), \rank(L), \rank(R)\}\\
    & \codefn 2 \cdot \max \{\decwidth(d), \rank(L), \rank(R)\}
    \qedhere
  \end{align*}
\end{proof}

From \Cref{prop:mwd-rwd-upper-bound}, \Cref{prop:mwd-rwd-lower-bound} and \Cref{prop:rec-rank-width-equivalent}, we obtain the main result of this section.

\begin{thm}\label{th:mwd-rwd}
  Let \(G\) be a graph and let \(g = \fullboundariesgraph{G}{\initmap{}}{\initmap{}}{\emptymat}{\emptymat}\) be the corresponding morphism in \(\propGraph\).
  Then, \(\frac{1}{2} \cdot \rankwidth(G) \leq \mwd(g) \leq 2 \cdot \rankwidth(G)\).
\end{thm}

\section{Conclusion and future work}
We defined monoidal width for measuring the complexity of morphisms in monoidal categories.
The concrete examples that we aimed to capture are tree width and rank width.
In fact, we have shown that, by choosing suitable categorical algebras, monoidal width is equivalent to these widths.
We have also related monoidal width to the rank of matrices over the natural numbers.

Our future goal is to leverage the generality of monoidal categories to study other examples outside the graph theory literature.
In the same way Courcelle's theorem gives fixed-parameter tractability of a class of problems on graphs with parameter tree width or rank width, we aim to obtain fixed-parameter tractability of a class of problems on morphisms of monoidal categories with parameter monoidal width.
This result would rely on Feferman-Vaught-Mostowski type theorems specific to the operations of a particular monoidal category \(\cat{C}\) or particular class of monoidal categories, which would ensure that the problems at hand respect the compositional structure of these categories.

\noindent\textbf{Conjecture.}
Computing a compositional problem on the set of morphisms \(\cat{C}_k(X,Y)\) with \(k\)-bounded monoidal width with a compositional algorithm is linear in \(\nodeweight\).
Explicitly, computing the solution on \(f \in \cat{C}_{k}(X,Y)\) takes \(\mathcal{O}(c(k) \cdot \nodeweight(f))\), for some more than exponential function \(c \colon \naturals \to \naturals\).

\bibliographystyle{alphaurl}
\bibliography{main.bib}
\end{document}